\documentclass[11pt]{article}

\usepackage[margin=1in]{geometry}
\usepackage[T1]{fontenc}
\usepackage[utf8]{inputenc}
\usepackage{lmodern}
\usepackage{microtype}
\usepackage{amsmath,amssymb,amsthm,mathtools}
\usepackage{booktabs,array}
\usepackage{tikz}
\usetikzlibrary{arrows.meta,positioning}
\usepackage{enumitem}
\usepackage{xcolor}
\usepackage[hidelinks]{hyperref}
\usepackage[nameinlink,noabbrev]{cleveref}

\allowdisplaybreaks

\newtheorem{theorem}{Theorem}[section]
\newtheorem{lemma}[theorem]{Lemma}
\newtheorem{proposition}[theorem]{Proposition}
\newtheorem{corollary}[theorem]{Corollary}
\theoremstyle{definition}
\newtheorem{definition}[theorem]{Definition}

\theoremstyle{remark}
\newtheorem{remark}[theorem]{Remark}

\newcommand{\F}{\mathbb{F}}
\newcommand{\Q}{\mathbb{Q}}
\newcommand{\N}{\mathbb{N}}
\newcommand{\AC}{\mathsf{AC}}
\newcommand{\PAC}{\mathsf{PAC\!\!-\!AC}}
\newcommand{\VPe}{\mathsf{VP}_{\!e}}
\newcommand{\CERS}{\mathsf{CERS}}
\newcommand{\CEEA}{\mathsf{CEEA}}
\newcommand{\psc}{\operatorname{psc}}
\newcommand{\lc}{\operatorname{lc}}
\newcommand{\DegRed}{\preceq_{\mathrm{DEG}}}
\newcommand{\eps}{\varepsilon}
\newcommand{\Match}{\operatorname{Match}}

\title{
The Complete Extended Euclidean Scheme Is Not in
Piecewise Arithmetic $\mathrm{AC}^0$\\
\large
A Continuant--Hankel Degeneration and a Fixed-Bound Subresultant Reduction
}

\author{
Amiel Ferman\thanks{
Email: \texttt{[amiel.ferman@openu.ac.il]}.
}
}
\date{August 1, 2026}

\begin{document}
\maketitle

\begin{abstract}
We prove that the complete extended Euclidean scheme for pairs of monic univariate polynomials over a field of characteristic zero cannot be computed by polynomial-size, 
constant-depth piecewise arithmetic circuits in the select-gate model of Andrews and Wigderson. In fact, the lower bound already holds for the simpler task of 
outputting the complete padded list of nonzero Euclidean remainders.

The proof connects the Euclidean algorithm to two classical polynomial families: continuants and consecutive Hankel determinants. 
We first show that the continuant family cannot be computed by polynomial-size arithmetic circuits of constant depth. 
This follows by combining the universality of the continuant for arithmetic formulas with the adjacent-depth hierarchy for arithmetic circuits. 
We then give an explicit degeneration that extracts the continuant from a consecutive Hankel determinant. Combinatorially, the lowest-weight terms in the determinant 
expansion correspond exactly to matchings of a path, whose matching polynomial is the continuant. This transfers the constant-depth lower bound from continuants to Hankel determinants.

Finally, we show that a suitable Hankel determinant can be recovered from fixed coordinates of the complete Euclidean remainder sequence on a nonempty Zariski-open set. 
The connection is provided by a middle principal subresultant coefficient. A generic removal of select gates, followed by constant-depth division elimination, 
would therefore turn any piecewise constant-depth algorithm for the complete remainder sequence into an ordinary constant-depth circuit for Hankel determinants, 
contradicting the lower bound above.

We also show that the same obstruction applies to several related outputs. 
It yields lower bounds for the complete polynomial continued-fraction expansion 
and for the complete profile of fixed-bound principal subresultant coefficients, 
since each of these outputs directly exposes the Hankel determinant used in the Euclidean reduction. 
In addition, we obtain a lower bound for normalized subdiagonal Pad'e approximation: even the normalized denominator alone suffices, 
through polynomially many parallel Pad'e computations and a telescoping product of determinantal ratios, 
to recover the same consecutive Hankel determinant. 
Consequently, none of these problems can be computed by polynomial-size, 
constant-depth piecewise arithmetic circuits.
\end{abstract}


\medskip

\newpage
\tableofcontents

\section{Introduction and main result}\label{sec:introduction}

For polynomials $f,g\in\F[x]$ with $\deg f>\deg g$, the ordinary Euclidean
algorithm constructs
\[
R_0=f,\qquad R_1=g,\qquad
R_{i-1}=Q_iR_i+R_{i+1},\qquad \deg R_{i+1}<\deg R_i,
\]
until the next remainder is zero.  Starting from
\[
(U_0,V_0)=(1,0),\qquad (U_1,V_1)=(0,1),
\]
the standard extended recurrence
\[
U_{i+1}=U_{i-1}-Q_iU_i,\qquad
V_{i+1}=V_{i-1}-Q_iV_i
\]
produces the canonical B\'ezout identities
\[
U_if+V_ig=R_i.
\]
The complete extended Euclidean scheme asks for the entire sequence of
quotients, nonzero remainders, and canonical B\'ezout coefficients, rather
than only the greatest common divisor and the final B\'ezout identity.

Classical parallel algorithms place this complete scheme in arithmetic
$\mathsf{NC}^2$ \cite{vzGathen84}.  More recently, Andrews and Wigderson gave
piecewise arithmetic $\AC^0$ algorithms for polynomial GCD, resultants,
discriminants, final B\'ezout coefficients, polynomial division, and several
related structured linear-algebraic problems \cite{AW}.  They left open
whether the complete extended Euclidean scheme can also be computed in
constant depth \cite[Section~11]{AW}.  We prove that it cannot.

Throughout the paper, $\AC^0_{\F}$ denotes ordinary division-free arithmetic
circuits of polynomial size and constant depth.  The notation $\PAC^0_{\F}$
denotes the piecewise select-gate model of Andrews and Wigderson; such circuits
may use divisions, subject to the requirement that every division be defined
on the declared input domain.

\begin{theorem}[Main theorem]\label[theorem]{thm:main}
Let $\F$ be a field of characteristic zero.  The family computing the complete
extended Euclidean scheme for pairs of monic univariate polynomials is not in
$\PAC^0_{\F}$.  The conclusion remains true if the required output is only the
padded complete list of ordinary nonzero Euclidean remainders.
\end{theorem}

We use two closely related computational problems.  The \emph{complete
Euclidean remainder scheme} (\(\CERS\)) takes a pair of input polynomials
\(f,g\) and outputs the entire ordinary Euclidean remainder sequence
\[
R_0=f,\qquad R_1=g,\qquad
R_{i-1}=Q_iR_i+R_{i+1},
\qquad \deg R_{i+1}<\deg R_i,
\]
with the output padded by zero polynomials to a fixed length determined by
the input degree bounds.  The \emph{complete extended Euclidean scheme}
(\(\CEEA\)) outputs, in addition, the complete quotient sequence and the
Bézout coefficient sequences \(U_i,V_i\) satisfying
\[
U_i f+V_i g=R_i
\]
at every step.  Thus \(\CERS\) is obtained from \(\CEEA\) simply by
discarding the quotient and Bézout-coordinate outputs.  Consequently, 
any upper bound for CEEA immediately gives the same upper
bound for CERS, whereas a lower bound for CERS also applies to CEEA.
Our result is therefore stronger than a lower bound for the full
extended output: it already holds for the remainder coordinates alone
and hence, a fortiori, for the extended Euclidean scheme in the output
convention of Andrews and Wigderson~\cite{AW}.

\subsection*{The proof at a glance}

The proof starts from the \emph{adjacent-depth hierarchy} for arithmetic
circuits.  Intuitively, this hierarchy says that constant-depth arithmetic
computation does not collapse: allowing even one additional layer of
arithmetic gates can make it possible to compute explicit polynomial families
that cannot be computed by polynomial-size circuits at the preceding depth.
More formally, Limaye, Srinivasan, and Tavenas proved that, for every fixed
integer \(D\geq 2\), there is an explicit polynomial family computable by
polynomial-size arithmetic circuits of depth \(D\), in their
linear-combination-gate model, but not by polynomial-size circuits of depth
\(D-1\) \cite{LST}.  After the elementary translation
between their model and the \(+,\times\) model used here, this provides, for
every fixed depth bound, a polynomial family of slightly larger constant depth
that cannot be compressed to that bound.  The argument transfers this
hierarchy-based hardness first to the continuant, then to consecutive Hankel
determinants, and finally to the complete Euclidean schemes.  The technical
sections make each transfer effective and verify that it preserves polynomial
size and constant depth.  We now describe the three steps.

\paragraph{Effective degeneration.}
The reduction used in the first two steps is denoted by
$\DegRed$.  Informally,
\[
f\DegRed g
\]
means that $f$ can be recovered as the constant term of $g$ after substituting
suitable affine Laurent forms depending on an auxiliary parameter $\eps$.
More precisely, the substitution has the form
\[
g\bigl(\ell_1(\eps,x),\ldots,\ell_t(\eps,x)\bigr)
 = f(x)+\eps h_1(x)+\cdots+\eps^E h_E(x),
\]
where the number of substituted forms, the error degree $E$, the Laurent
exponents, and the descriptions of the forms are all polynomially bounded.
The full definition is given in \cref{sec:degenerations}.

The point of this relation is computational.  If $g$ has polynomial-size
circuits of depth $d$, then one may evaluate the displayed identity at
$E+1$ nonzero values of $\eps$ and recover its constant term by interpolation:
choose \(E+1\) distinct nonzero elements 
\[ \alpha_0,\ldots,\alpha_E\in F. \] 
Lagrange interpolation provides field constants 
\(\lambda_0,\ldots,\lambda_E\) such that, 
for every polynomial \(P(\varepsilon)\) of degree at most \(E\), 
\[ [\varepsilon^0]P(\varepsilon) = \sum_{j=0}^{E}\lambda_j P(\alpha_j). \] 
Consequently, \[ f(x) = \sum_{j=0}^{E} \lambda_j\, g\bigl( \ell_1(\alpha_j,x),\ldots,\ell_t(\alpha_j,x) \bigr). \] 
The \(E+1\) specialized copies of the circuit for \(g\) are evaluated in parallel. 
Since addition gates have unbounded fan-in and field constants are freely available in the nonuniform model, 
interpolation adds only a constant number of arithmetic layers and polynomial size. Thus an effective degeneration 
turns a depth-\(d\) circuit for \(g\) into a polynomial-size circuit for \(f\) of depth \(d+O(1)\).

Consequently, an effective degeneration
\[
f\DegRed g
\]
turns a depth-$d$ circuit for $g$ into a polynomial-size circuit for $f$ of
depth $d+O(1)$.  Thus constant-depth easiness passes from right to left, or,
equivalently, constant-depth hardness passes from left to right.

\paragraph{New reductions and external ingredients.}
The proof combines two new reductions established in this paper with
several results from the existing literature. The first new reduction
is the explicit effective Laurent degeneration
\[
K_n \preceq_{\mathrm{DEG}} \Delta_{n+1},
\]
proved in Theorem~4.1. Thus the continuant \(K_n\) is recovered as the
constant coefficient, with respect to the auxiliary parameter
\(\varepsilon\), of a consecutive Hankel determinant after an explicit
affine Laurent substitution. In particular, this is not merely an
appeal to the general universality of the continuant: it is a direct
and quantitatively effective reduction from the continuant to the
specific Hankel family used later in the proof.

The second new reduction is the Euclid-to-Hankel implication proved in
Theorem~1.3:
\[
\mathrm{CERS}\in \mathrm{PAC}\text{-}AC^0_F
\quad\Longrightarrow\quad
(\Delta_N)_{N\geq 1}\in AC^0_F.
\]
It is obtained by identifying a middle fixed-bound principal
subresultant coefficient of the pair
\[
\left(x^{2n},
a_0x^{2n-1}+a_1x^{2n-2}+\cdots+a_{2n-2}x+1\right)
\]
with a fixed sign times \(\Delta_n\), and then reconstructing this
subresultant from fixed leading-coefficient coordinates of the
complete Euclidean remainder output on a nonempty normal locus.

The remaining ingredients are taken from earlier work. The
universality of the continuant for polynomial-size arithmetic
formulas, together with the matrix gadgets underlying its
quantitative form, comes from Bringmann, Ikenmeyer, and
Zuiddam~\cite{BIZ}. Formula balancing is supplied by
Brent~\cite{Brent}. The constant-depth separation used to derive the
continuant lower bound is the adjacent-depth hierarchy of Limaye,
Srinivasan, and Tavenas~\cite{LST}. Finally, the generic removal of
select gates and the constant-depth elimination of division gates use
the results of Andrews and Wigderson~\cite{AW}. Accordingly, the two
principal contributions of the present paper are the explicit
continuant--Hankel degeneration and the reduction from complete
Euclidean remainders to consecutive Hankel determinants; the cited
results provide the complexity-theoretic framework in which these two
reductions yield the main lower bound.

\paragraph{Note added.}
After completion of this manuscript, I learned that Andrews and Wigderson
had independently obtained the same main result. The two works were developed
independently. We plan to prepare a joint final version for submission to a
journal and conference.

\paragraph{Step 1: the continuant is not in arithmetic $\AC^0$.}
The continuant family is defined by
\[
K_0=1,\qquad K_1(x_1)=x_1,\qquad
K_n(x_1,\ldots,x_n)=x_nK_{n-1}+K_{n-2}.
\]
Its relevance comes from its universality for the formula class $\VPe$.
Informally, $\VPe$ consists of polynomial families that can be computed by
division-free arithmetic formulas of polynomial size; unlike circuits,
formulas do not allow a previously computed value to be reused by several
later gates.  The formal definition appears in \cref{def:vpe}.

Bringmann, Ikenmeyer, and Zuiddam \cite{BIZ} show that every family in this class can be
obtained from a polynomially larger continuant by an appropriate degeneration.
In the effective form established in
\cref{sec:continuant-simulation}, every $(f_n)\in\VPe$ satisfies
\[
(f_n)\DegRed(K_n),
\]
up to a polynomially bounded reindexing.  Thus, if the continuant had
polynomial-size circuits of some fixed depth $d$, the preservation property of
effective degenerations would imply that every family in $\VPe$ is computable
by polynomial-size circuits of depth $d+O(1)$.

This conclusion contradicts the adjacent-depth hierarchy.  Indeed, every
polynomial-size circuit of fixed depth can be unfolded into a
polynomial-size formula, and hence computes a family in $\VPe$.  The hierarchy
theorem of Limaye, Srinivasan, and Tavenas\cite{LST} supplies fixed-depth
polynomial-size circuit families that cannot be computed at the smaller
constant depth forced by the assumed upper bound for the continuant.
Consequently,
\[
(K_n)\notin\AC^0_{\F}.
\]

\paragraph{Step 2: the continuant occurs inside a Hankel determinant.}
For $N\geq1$, let
\[
\Delta_N(a_0,\ldots,a_{2N-2})
 :=\det(a_{i+j})_{0\leq i,j<N}
\]
be the consecutive Hankel determinant.  Our first new reduction is the
explicit effective degeneration
\[
K_n\DegRed\Delta_{n+1}.
\]
The substitution assigns carefully chosen powers of $\eps$ to the Hankel
anti-diagonals.  In the determinant expansion, the terms of lowest weight are
exactly the permutations that are products of disjoint adjacent
transpositions.  These permutations are naturally identified with matchings
of a path, and the corresponding matching polynomial is precisely $K_n$.

If the Hankel determinant were in $\AC^0_{\F}$, the degeneration would put the
continuant in $\AC^0_{\F}$ as well.  Step~1 therefore implies the following
intermediate lower bound.

\begin{theorem}[Hankel lower bound]\label[theorem]{thm:hankel-lb}
Over every field of characteristic zero, the family $(\Delta_N)_{N\geq1}$ is
not in arithmetic $\AC^0_{\F}$.
\end{theorem}

\paragraph{Step 3: complete Euclidean remainders would compute Hankel
determinants.}
The second new reduction connects $\Delta_n$ to the Euclidean algorithm.  For
the special pair
\[
F_n(x)=x^{2n},\qquad
G_a(x)=a_0x^{2n-1}+a_1x^{2n-2}+\cdots+a_{2n-2}x+1,
\]
a middle principal subresultant coefficient is, up to a fixed sign, exactly
\[
\Delta_n(a_0,\ldots,a_{2n-2}).
\]
On a nonempty Zariski-open set, the ordinary Euclidean remainder sequence of
this pair is normal: its degree drops by exactly one at each step.  On this
normal locus, the same principal subresultant is a simple product of the
leading coefficients of the successive remainders.  Those leading
coefficients occur at fixed output coordinates of the padded complete
remainder list.

It follows that a piecewise constant-depth circuit for the complete remainder
scheme would give a piecewise circuit agreeing with $\Delta_n$ on a nonempty
open set.  Generic removal of the select gates produces an ordinary rational
circuit computing $\Delta_n$ identically, and constant-depth division
elimination then produces a division-free $\AC^0_{\F}$ circuit.  We therefore
obtain the implication

\begin{theorem}[Euclid-to-Hankel implication]\label[theorem]{thm:euclid-hankel}
Over every field of characteristic zero,
\[
\CERS\in\PAC^0_{\F}
\quad\Longrightarrow\quad
(\Delta_N)_{N\geq1}\in\AC^0_{\F},
\]
where $\CERS$ denotes the complete ordinary Euclidean remainder scheme for
monic inputs.
\end{theorem}

Combining the Hankel lower bound with the Euclid-to-Hankel implication shows
that $\CERS\notin\PAC^0_{\F}$.  Since the complete remainder list is a
coordinate projection of the complete extended Euclidean output, the same
lower bound holds for the complete extended scheme.

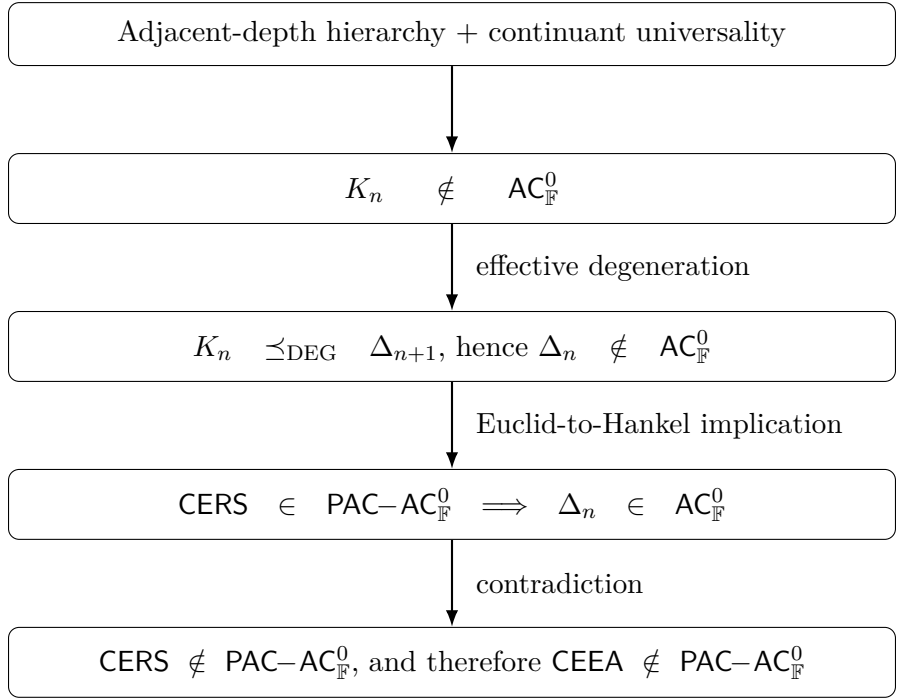
\begin{figure}[ht]
\centering
\begin{tikzpicture}[
  node distance=1.15cm,
  proofbox/.style={
    draw,
    rounded corners,
    align=center,
    inner sep=7pt,
    text width=0.68\linewidth
  },
  proofarrow/.style={-{Latex[length=2.2mm]},thick}
]
\node[proofbox] (hierarchy)
  {Adjacent-depth hierarchy + continuant universality};
\node[proofbox,below=of hierarchy] (continuant)
  {$K_n\notin\AC^0_{\F}$};
\node[proofbox,below=of continuant] (hankel)
  {$K_n\DegRed\Delta_{n+1}$, hence $\Delta_n\notin\AC^0_{\F}$};
\node[proofbox,below=of hankel] (euclid)
  {$\CERS\in\PAC^0_{\F}\Longrightarrow\Delta_n\in\AC^0_{\F}$};
\node[proofbox,below=of euclid] (conclusion)
  {$\CERS\notin\PAC^0_{\F}$, and therefore
   $\CEEA\notin\PAC^0_{\F}$};

\draw[proofarrow] (hierarchy) -- (continuant);
\draw[proofarrow] (continuant) -- node[right=5pt] {effective degeneration} (hankel);
\draw[proofarrow] (hankel) -- node[right=5pt,align=left] {Euclid-to-Hankel\ implication} (euclid);
\draw[proofarrow] (euclid) -- node[right=5pt] {contradiction} (conclusion);
\end{tikzpicture}
\caption{The logical flow of the proof.  The first two steps establish an
ordinary arithmetic $\AC^0$ lower bound for consecutive Hankel determinants;
the last step transfers this lower bound to the complete Euclidean schemes.}
\label{fig:proof-overview}
\end{figure}

\section{Computational models and output conventions}\label{sec:models}

\subsection{Rational and division-free arithmetic circuits}

\begin{definition}[Rational arithmetic circuit]\label[definition]{def:rational-circuit}
A rational arithmetic circuit over $\F$ is a directed acyclic graph whose
input gates are variables or constants from $\F$, and whose internal gates are
labelled by $+$, $\times$, or division.  At the formal level, the rational
function used as the divisor at every division gate is required not to be
identically zero.  The size is the number of wires and the depth is the maximum
length of an input--output path.

When such a circuit is asserted to compute a function on a set
$D\subseteq\F^N$, every division gate is additionally required to be defined at
every point of $D$.
\end{definition}

For polynomial families, $\AC^0_{\F}$ denotes polynomial-size,
constant-depth, division-free arithmetic circuits with unbounded fan-in
addition and multiplication.  Constants from $\F$ are nonuniformly available,
and subtraction is implemented using the constant $-1$.

We will temporarily use rational circuits.  Strassen's original
division-elimination theorem~\cite{Strassen} shows how divisions can be
removed from arithmetic computations of polynomials.  We use the
depth-preserving constant-depth formulation proved by Andrews and
Wigderson~\cite[Theorem~2.13]{AW}: a rational
$\mathrm{AC}^{0}$ circuit computing a polynomial of polynomially bounded
degree can be converted into a division-free $\mathrm{AC}^{0}$ circuit
with polynomial size and constant depth.  In the same work, polynomial
division with remainder is given in
\cite[Theorem~2.12]{AW}.

\subsection{Piecewise arithmetic circuits: the Andrews--Wigderson model}

For $m\geq1$, define
\[
\operatorname{select}(y_1,\ldots,y_m)=
\begin{cases}
y_1,&y_1\ne0,\\
y_2,&y_1=0,\ y_2\ne0,\\
\ \vdots\\
y_m,&y_1=\cdots=y_{m-1}=0,\ y_m\ne0,\\
0,&y_1=\cdots=y_m=0.
\end{cases}
\]

\begin{definition}[Piecewise arithmetic circuit]\label[definition]{def:piecewise}
A piecewise arithmetic circuit over $\F$ is a rational arithmetic circuit
augmented by internal gates labelled $\operatorname{select}$, with arbitrary
fan-in.  As in Andrews and Wigderson's \cite[Definitions~2.3]{AW}, no division by zero is
permitted on any input in the declared domain.  The size and depth are defined
as for ordinary arithmetic circuits.  A family belongs to $\PAC^0_{\F}$ if it
is computed by polynomial-size, constant-depth piecewise arithmetic circuits.
\end{definition}

This is exactly the internal select-gate model of
\cite[Definitions~2.3 and~2.4]{AW}.  In intermediate reductions we will also
speak of a \emph{formal piecewise rational circuit on a nonempty open set
$U$}: this means that every division gate is defined at every point of $U$,
without assigning any value outside $U$.  This restricted-domain usage is only
an intermediate device; the assumed circuit for $\CERS$ is a genuine
piecewise circuit on its full monic-input domain.

Ordinary composition and coordinate projection preserve polynomial size and
constant depth.  The same is true in the select-gate model.  In particular, 
substituting a rational expression such as \(a_0^{-1}\)
produces a formal piecewise rational circuit on the nonempty Zariski-open set
\[
\{a\in F^N : a_0\neq 0\}.
\]
We denote this set by \(D(a_0)\).

\subsection{Generic extraction on a prescribed open set}

\begin{lemma}[Nonvanishing over an infinite field]\label[lemma]{lem:nonvanishing}
Let $\F$ be infinite.  If $p\in\F[z_1,\ldots,z_N]$ is nonzero, then some
$a\in\F^N$ satisfies $p(a)\ne0$.  Consequently, a finite intersection of
nonempty basic Zariski-open subsets of $\F^N$ is nonempty.
\end{lemma}

\begin{proof}
The first assertion is proved by induction on $N$.  For $N=1$, a nonzero
polynomial has only finitely many roots.  For $N>1$, write
\[
p=\sum_{k=0}^d p_k(z_1,\ldots,z_{N-1})z_N^k
\]
with some $p_k\ne0$.  By induction, specialize $z_1,\ldots,z_{N-1}$ so that
this coefficient remains nonzero.  The resulting univariate polynomial is
nonzero and therefore does not vanish on all of $\F$.

A basic open set has the form
\[
D(q):=\{z\in\F^N:q(z)\ne0\}.
\]
If $D(q_1),\ldots,D(q_t)$ are nonempty, then every $q_i$ is nonzero and
\[
\bigcap_{i=1}^tD(q_i)=D(q_1\cdots q_t)
\]
is nonempty by the first assertion.
\end{proof}

Every nonempty Zariski-open $U\subseteq\F^N$ contains a nonempty basic open
set.  Indeed, if $a\in U$ and $\F^N\setminus U=V(I)$, then some $u\in I$
satisfies $u(a)\ne0$, and $D(u)\subseteq U$.

The next lemma is the prescribed-open-set variant of the select-removal
argument in \cite[Lemma~2.5]{AW}.  The extra domain bookkeeping is needed
because the monic normalization used later is defined only on $D(a_0)$.

\begin{lemma}[Generic extraction of an ordinary circuit]\label[lemma]{lem:generic-extraction}
Let $\F$ be infinite.  Let $\Phi$ be a scalar-output formal piecewise rational
circuit on a nonempty Zariski-open set $U\subseteq\F^N$.  Suppose that the
function computed by $\Phi$ agrees on $U$ with a rational function
$h\in\F(z_1,\ldots,z_N)$.  Then there is an ordinary rational arithmetic
circuit $C$, obtained by replacing every select gate by one of its inputs or by
the constant zero, such that $C=h$ identically in
$\F(z_1,\ldots,z_N)$.  The size and depth of $C$ do not exceed those of
$\Phi$.
\end{lemma}

\begin{proof}
List the select gates of $\Phi$ as $v_1,\ldots,v_s$ in a topological order.
We construct circuits
\[
\Phi=\Phi^{(0)},\Phi^{(1)},\ldots,\Phi^{(s)}
\]
and nonempty open sets
\[
U=U_0\supseteq U_1\supseteq\cdots\supseteq U_s
\]
so that after stage $q$: the first $q$ select gates have been removed; every
gate of $\Phi^{(q)}$ is defined on $U_q$; and $\Phi^{(q)}$ agrees with $\Phi$
on $U_q$.

Suppose the construction is complete through stage $q-1$, and write
\[
v_q=\operatorname{select}(w_1,\ldots,w_m).
\]
Every select gate in the child subcircuits has already been removed, so $w_j$
computes an ordinary rational function $r_j$.

If every $r_j$ is the zero rational function, all children evaluate to zero on
$U_{q-1}$.  Replace $v_q$ by zero and take $U_q=U_{q-1}$.

Otherwise, let $j$ be the least index for which $r_j\ne0$, and write
\[
r_j=A_j/B_j,
\qquad A_j,B_j\in\F[z_1,\ldots,z_N],
\qquad A_jB_j\ne0.
\]
Choose a nonzero $u$ with $D(u)\subseteq U_{q-1}$ and put
\[
U_q:=D(uA_jB_j).
\]
By \cref{lem:nonvanishing}, this is a nonempty open subset of $U_{q-1}$.
For $k<j$, the function $r_k$ is identically zero, while $r_j$ is defined and
nonzero throughout $U_q$.  Thus on $U_q$ the select gate equals its child
$w_j$.  Replace $v_q$ by that child.  Values at all downstream gates are
unchanged on $U_q$, so their division gates remain defined there.

After all select gates have been removed, let $C=\Phi^{(s)}$ and $W=U_s$.
Then $C=h$ on the nonempty open set $W$.  Write
\[
C-h=A/B
\]
with $A,B\in\F[z_1,\ldots,z_N]$ and $B\ne0$, and choose a nonzero $u$ with
$D(u)\subseteq W$.  On $D(uB)$, one has $A=0$.  If $A\ne0$, then $uAB$ is a
nonzero polynomial and \cref{lem:nonvanishing} supplies a point of $D(uB)$ at
which $A\ne0$, a contradiction.  Therefore $A=0$ and $C=h$ identically.

Each replacement deletes a select gate and uses an existing child subcircuit
or the constant zero.  Hence neither size nor depth increases.
\end{proof}

\begin{remark}
The extraction lemma is existential. It does not provide an 
identity-testing algorithm for choosing the nonzero child of a select gate.
\end{remark}

\subsection{The output conventions}\label{subsec:remainder-output}

Fix a degree bound $m$.  Every polynomial is encoded by its coefficient vector
in the ordered basis
\[
x^m,x^{m-1},\ldots,1,
\]
with leading zeros when its actual degree is smaller.

A complete ordinary remainder output contains exactly $m+1$ polynomial slots.
The nonzero remainders $R_0,\ldots,R_t$ occupy the first $t+1$ slots; the
terminal zero remainder is not stored; and all remaining slots are zero.  On a
locus where $\deg R_i=m-i$, the leading coefficient $\lc(R_i)$ is therefore a
fixed coordinate of the $i$th slot.  Any input-independent rearrangement or
zero-padding convention is interconvertible with this one by depth-one
coordinate maps.

\begin{definition}[Complete remainder and extended-Euclidean families]\label[definition]{def:cers-family}
For integers $m>r\geq0$, let $\mathcal M_{m,r}$ consist of the monic pairs
\[
\left(
 x^m+f_1x^{m-1}+\cdots+f_m,
 x^r+g_1x^{r-1}+\cdots+g_r
\right).
\]
For $(f,g)\in\mathcal M_{m,r}$, let
\[
R_0=f,\qquad R_1=g,\qquad
R_{i-1}=Q_iR_i+R_{i+1},\qquad \deg R_{i+1}<\deg R_i,
\]
be the ordinary Euclidean sequence, terminated immediately before the first
zero remainder. For a polynomial $S$ of degree at most $m$, let
\[
\operatorname{coeff}_m(S)
  :=\bigl([x^m]S,[x^{m-1}]S,\ldots,[x^0]S\bigr)\in F^{m+1}.
\]
For any list of polynomials $S_1,\ldots,S_v$, where $v\leq m+1$
and every $S_i$ has degree at most $m$, define
\[
\operatorname{Pad}_m(S_1,\ldots,S_v)
 :=
 \bigl(
   \operatorname{coeff}_m(S_1),\ldots,
   \operatorname{coeff}_m(S_v),
   \underbrace{\operatorname{coeff}_m(0),\ldots,
               \operatorname{coeff}_m(0)}_{m+1-v\text{ slots}}
 \bigr).
\]

Write the nonzero Euclidean remainders as
\[
R_0,R_1,\ldots,R_t.
\]
Thus the computed quotients are $Q_1,\ldots,Q_t$, with
\[
R_{i-1}=Q_iR_i+R_{i+1}
        \qquad (1\leq i<t),
\]
and the final exact division is
\[
R_{t-1}=Q_tR_t.
\]
The complete remainder map is
\[
\mathrm{CERS}_{m,r}(f,g)
   :=\operatorname{Pad}_m(R_0,\ldots,R_t).
\]

For the extended output, define
\[
(U_0,V_0)=(1,0),
\qquad
(U_1,V_1)=(0,1),
\]
and, for $1\leq i<t$, define
\[
U_{i+1}=U_{i-1}-Q_iU_i,
\qquad
V_{i+1}=V_{i-1}-Q_iV_i.
\]
Then
\[
U_if+V_ig=R_i
\qquad (0\leq i\leq t).
\]
The complete extended-Euclidean map is the concatenation
\[
\begin{split}
\mathrm{CEEA}_{m,r}(f,g):={}&
\Bigl(
 \operatorname{Pad}_m(Q_1,\ldots,Q_t),\\
&\operatorname{Pad}_m(R_0,\ldots,R_t),\\
&\operatorname{Pad}_m(U_0,\ldots,U_t),\\
&\operatorname{Pad}_m(V_0,\ldots,V_t)
\Bigr).
\end{split}
\]
Here each polynomial is represented by its coefficient vector in the
fixed ambient basis
\[
x^m,x^{m-1},\ldots,1.
\]

We record the degree bounds that make all four applications of
$\operatorname{Pad}_m$ well defined. Put
\[
d_i:=\deg R_i\qquad (0\leq i\leq t).
\]
Since the remainder degrees decrease strictly, one has $t\leq m$, so
each displayed list contains at most (m+1) polynomials. 

Moreover,
\[
\begin{aligned}
\deg Q_i
  &= d_{i-1}-d_i \leq m
  && (1\leq i\leq t),\\
\deg U_i
  &\leq d_1-d_{i-1}\leq m
  && (2\leq i\leq t),\\
\deg V_i
  &\leq d_0-d_{i-1}\leq m
  && (1\leq i\leq t).
\end{aligned}
\]
These bounds follow directly by induction from the extended Euclidean
recurrences, with the convention \(\deg 0=-\infty\).

The initial polynomials $U_0,U_1,V_0,V_1$ also have degree at most
$m$, and of course every remainder $R_i$ has degree at most $m$.
Thus every polynomial occurring in the quotient, remainder, and
Bézout-coefficient lists lies in the ambient coefficient space used by
$\operatorname{Pad}_m$. With this convention, $\mathrm{CERS}_{m,r}$ is obtained from
$\mathrm{CEEA}_{m,r}$ by a fixed coordinate projection.

We write $\CERS\in\PAC^0_{\F}$ (and analogously for $\CEEA$) if constants
$c,D$ exist, independent of $m$ and $r$, such that every corresponding map has
a piecewise circuit of size at most $m^c$ and depth at most $D$.
\end{definition}

\section{Continuant Universality and Constant-Depth Lower Bounds}\label{sec:Contin}

\subsection{Formula families and the continuant}

\begin{definition}[The formula class $\VPe$]\label[definition]{def:vpe}
A polynomial family $(f_n)_{n\geq1}$ belongs to $\VPe$ if some polynomial
$p\in\N[X]$ bounds the size of a fan-in-two, division-free arithmetic formula \(\Phi_n\)
for $f_n$ by $p(n)$.  
The leaves of \(\Phi_n\), which are its input nodes, are labelled by variables or constants from \(F\); 
every internal node is labelled by \(+\) or \(\times\); and the value computed at the unique root is the output polynomial \(f_n\).
\end{definition}

All polynomial bounds in this paper refer to the family index.  In the
nonuniform algebraic model, a field constant is an atomic label of unit
syntactic cost.

\begin{definition}[Continuant]\label[definition]{def:continuant}
The continuant polynomials are defined by
\[
K_0=1,\qquad K_1(x_1)=x_1,
\]
and, for $n\geq2$,
\[
K_n(x_1,\ldots,x_n)
=x_nK_{n-1}(x_1,\ldots,x_{n-1})
 +K_{n-2}(x_1,\ldots,x_{n-2}).
\]
\end{definition}

It turns out that the continuant also admits a useful combinatorial
description in terms of matchings of a path. Let \[ P_n=1-2-\cdots-n \] be the path graph on the vertex set \(\{1,\ldots,n\}\). 
Denote by \(\Match(P_n)\) the set of all matchings of \(P_n\), that is, all sets \(M\) of pairwise vertex-disjoint edges of \(P_n\). For \(M\in\Match(P_n)\), let \(V(M)\) be the set of vertices incident with the edges of \(M\). Then \[ K_n(x_1,\ldots,x_n) = \sum_{M\in\Match(P_n)} \prod_{i\notin V(M)} x_i. \] 
Thus each matching contributes the product of the variables associated with the vertices left uncovered by that matching.

Also, for
\[
Q(z)=\begin{pmatrix}z&1\\1&0\end{pmatrix},
\]
a direct induction gives
\begin{equation}\label{eq:continuant-matrix}
K_n(z_1,\ldots,z_n)
=\bigl(Q(z_n)\cdots Q(z_1)\bigr)_{1,1}.
\end{equation}

\subsection{Effective Laurent degenerations}\label{sec:degenerations}

An \emph{affine Laurent form} in $x=(x_1,\ldots,x_m)$ is
\[
\ell(\eps,x)=c_0(\eps)+\sum_{j=1}^m c_j(\eps)x_j,
\qquad c_j(\eps)\in\F[\eps,\eps^{-1}].
\]

\begin{definition}[Effective polynomial-error Laurent degeneration]\label[definition]{def:degeneration}
For $f\in\F[x_1,\ldots,x_m]$ and $g\in\F[y_1,\ldots,y_t]$, write
$f\DegRed g$ if affine Laurent forms $\ell_1,\ldots,\ell_t$ and an integer
$E\geq0$ satisfy
\begin{equation}\label{eq:deg-def}
g(\ell_1(\eps,x),\ldots,\ell_t(\eps,x))
=f(x)+\sum_{r=1}^{E}\eps^rh_r(x)
\end{equation}
for polynomials $h_r\in\F[x]$.  Thus the substituted expression contains no
negative power of $\eps$, and its constant coefficient is $f$.

For polynomial families \((f_n)_{n\geq 1}\) and
\((g_m)_{m\geq 1}\), we write
\[
(f_n)\preceq_{\mathrm{DEG}}(g_m)
\]
if there exists a function \(m\colon\mathbb{N}\to\mathbb{N}\),
bounded by a polynomial in \(n\), such that for every \(n\),
\[
f_n\preceq_{\mathrm{DEG}}g_{m(n)},
\]
and the number of substituted forms, the error degree, the largest
absolute Laurent exponent, and the total syntactic description length
of the forms are all polynomially bounded in \(n\).
\end{definition}

\medskip
\noindent\textbf{Nonuniformity convention.}
The adjective \emph{effective} is used here in the nonuniform
syntactic sense appropriate to algebraic circuit complexity.  Namely,
for every family index $n$, the required Laurent forms and the
reindexing are allowed to depend on $n$, but their number, exponents,
error degree, and total descriptions must satisfy the polynomial bounds
specified above.  Definition~3.3 does not, by itself, assert the
existence of a uniform algorithm that produces these descriptions from
$n$.  The nonuniform formulation is sufficient for all circuit
reductions used in this paper.
\medskip

The computational motivation for this definition is that an effective
polynomial-error Laurent degeneration behaves as a constant-depth reduction.
Indeed, suppose that
\[
(f_n)\preceq_{\mathrm{DEG}}(g_n)
\]
and that \(g_{m(n)}\) is computed by polynomial-size circuits of fixed depth.
Given an input \(x\) for \(f_n\), substitute the prescribed Laurent forms into
a circuit for \(g_{m(n)}\).  After assigning any fixed nonzero value to
\(\varepsilon\), all Laurent coefficients become field constants, so the
specialized expression can be evaluated with only a constant increase in
depth.  Since the resulting polynomial in \(\varepsilon\) has polynomially
bounded degree, its constant coefficient \(f_n(x)\) can be recovered from
polynomially many such evaluations by interpolation.  Thus an efficient
constant-depth circuit for \(g\) yields one for \(f\), with only polynomial
size overhead and constant additional depth.  Equivalently, a lower bound for
\(f\) transfers through the degeneration to a lower bound for \(g\). 
We prove this formally in \cref{lem:degeneration-preserves-depth}.

\subsection{Quantitative continuant simulation}\label{sec:continuant-simulation}

Bringmann, Ikenmeyer, and Zuiddam prove that the continuant is complete
for formula families under polynomially bounded degenerations
\cite[Propositions~3.5 and~3.6 and Theorem~3.12]{BIZ}.
The following formulation records the explicit quantitative bounds in their
construction, together with the Laurent-exponent bound required for the
effective degeneration used here.

For an integer $k\geq 0$, we write
\[
B=O(\varepsilon^k)
\]
for a matrix
\[
B\in
\varepsilon^k\operatorname{Mat}_{2\times 2}
\bigl(F[\varepsilon,\mathbf{x}]\bigr).
\]
Accordingly, for
\[
M\in
\operatorname{Mat}_{2\times 2}
\bigl(F[\varepsilon,\varepsilon^{-1},\mathbf{x}]\bigr),
\]
the notation
\[
M+O(\varepsilon^k)
\]
denotes a matrix of the form $M+B$ with
\[
B\in
\varepsilon^k\operatorname{Mat}_{2\times 2}
\bigl(F[\varepsilon,\mathbf{x}]\bigr).
\]
In particular, the error matrix contains no negative powers of
$\varepsilon$. We say that the error has $\varepsilon$-degree at most
$E$ if it admits an expansion
\[
B=\sum_{r=k}^{E}\varepsilon^r B_r,
\qquad
B_r\in\operatorname{Mat}_{2\times 2}(F[\mathbf{x}]).
\]

As in \cite{BIZ}, we call each individual factor \(Q(\ell_i)\) 
\emph{primitive \(Q\)-matrix} or \emph{primitive factor}, and call
the affine Laurent form \(\ell_i\) its \emph{primitive label}.

\begin{remark}
The recursive \(Q\)-matrix gadgets used below are taken from Bringmann, Ikenmeyer, and Zuiddam~\cite{BIZ}.  
\Cref{prop:quant} supplements their construction with the quantitative bookkeeping needed here, 
namely polynomial bounds on the error degree, the Laurent exponents, and the total syntactic 
description length of the primitive labels.
\end{remark}

\begin{proposition}[Quantitative \(Q\)-matrix simulation]\label[proposition]{prop:quant}
Assume that $\operatorname{char}F\neq 2$. Let $f$ be computed by a
fan-in-two division-free arithmetic formula of depth $d$. Then there are
affine Laurent forms
\[
\ell_1(\varepsilon,\mathbf{x}),\ldots,
\ell_t(\varepsilon,\mathbf{x})
\]
such that
\[
Q(\ell_t)\cdots Q(\ell_1)
=
Q(f)+\sum_{r=1}^{E}\varepsilon^r A_r,
\qquad
A_r\in\operatorname{Mat}_{2\times 2}(F[\mathbf{x}]),
\]
where
\[
t\leq 45\cdot 9^d,
\qquad
E\leq 12\cdot 25^d.
\]
Moreover, every exponent of $\varepsilon$ occurring in a coefficient of
one of the forms $\ell_i$ has absolute value at most
\[
2\cdot 3^d,
\]
and the total syntactic description length of the list
$(\ell_1,\ldots,\ell_t)$ is polynomial in $45\cdot 9^d$ and $d$.
\end{proposition}

\begin{proof}
We use the recursive constructions of Bringmann, Ikenmeyer, and
Zuiddam~\cite[Lemmas~3.2--3.4 and Proposition~3.5]{BIZ}, while
keeping track of the Laurent exponents and of the syntactic
descriptions of the primitive factors.

We prove the following stronger assertion by induction on \(d\).
If a polynomial \(p\) is computed by a fan-in-two division-free
arithmetic formula of depth at most \(d\), then, for every
\(\alpha \in F\), there exist affine Laurent forms
\[
\ell_1(\varepsilon,x),\ldots,\ell_t(\varepsilon,x)
\]
such that
\[
Q(\ell_t)\cdots Q(\ell_1)
=
Q(\alpha p)
+
\sum_{r=1}^{e}\varepsilon^r A_r,
\qquad
A_r \in \operatorname{Mat}_{2\times 2}(F[x]),
\]
where
\[
t \leq t_d := 45\cdot 9^d,
\qquad
e \leq E_d := 12\cdot 25^d,
\]
and every exponent of \(\varepsilon\) occurring in a coefficient of
one of the forms \(\ell_i\) has absolute value at most
\[
L_d := 2\cdot 3^d.
\]
In particular, the error contains no negative powers of
\(\varepsilon\).

We also maintain the following syntactic invariant. Every primitive
label is obtained either from a leaf label of the original formula or
from one of the fixed labels appearing in the gadgets of
\cite[Lemmas~3.2--3.4]{BIZ}, possibly after repeatedly applying the
substitution
\[
\varepsilon \longmapsto \varepsilon^3.
\]
Consequently, every primitive label remains an affine Laurent form in
the original variables and has Laurent support bounded by an absolute
constant.

Suppose first that \(d=0\). Then \(p\) is either an input variable or
a field constant. The matrices
\[
Q(\alpha/2)
\qquad\text{and}\qquad
Q(p)
\]
are exact primitive \(Q\)-matrices. In particular, they may be viewed
as elements of
\[
Q(\alpha/2)+O(\varepsilon^3)
\qquad\text{and}\qquad
Q(p)+O(\varepsilon^3),
\]
respectively. Since \(\operatorname{char} F\neq 2\), the scalar
\(\alpha/2\) is defined.

Apply the multiplication gadget of
\cite[Lemma~3.4]{BIZ}, with the polynomial denoted there by \(f\)
equal to \(\alpha\), and the polynomial denoted there by \(g\) equal
to \(p\). Its first input is therefore an approximation to
\[
Q(f/2)=Q(\alpha/2),
\]
as required. The gadget produces a product satisfying
\[
Q(\ell_t)\cdots Q(\ell_1)
=
Q(\alpha p)+O(\varepsilon).
\]
The number of primitive factors is at most
\[
4\cdot 1+4\cdot 1+37
=
45
=
t_0,
\]
and the error degree is at most
\[
4\cdot 0+4\cdot 0+12
=
12
=
E_0.
\]

The additional primitive labels introduced by the gadget are fixed
affine Laurent expressions involving only
\[
0,\quad
\pm 1,\quad
\pm\varepsilon,\quad
\pm\varepsilon^{-1},\quad
\pm\varepsilon^2,\quad
\varepsilon-1,\quad
\varepsilon^{-1}-1,
\]
and scalar multiples of these expressions. Thus every Laurent
exponent has absolute value at most
\[
2=L_0.
\]
This proves the induction assertion for \(d=0\).

Now assume that \(d\geq 1\). The root of the formula computing \(p\)
is either an addition gate or a multiplication gate.

Suppose first that
\[
p=g+h,
\]
where \(g\) and \(h\) are computed by formulas of depth at most
\(d-1\). Apply the induction hypothesis to the pairs
\((g,\alpha)\) and \((h,\alpha)\). We obtain products
\[
G
=
Q(\alpha g)
+
\sum_{r=1}^{e_g}\varepsilon^r B_r
\]
and
\[
H
=
Q(\alpha h)
+
\sum_{r=1}^{e_h}\varepsilon^r C_r,
\]
where
\[
B_r,C_r \in \operatorname{Mat}_{2\times 2}(F[x]),
\]
and
\[
t_g,t_h \leq 45\cdot 9^{d-1},
\qquad
e_g,e_h \leq 12\cdot 25^{d-1}.
\]

By the addition gadget of \cite[Lemma~3.2]{BIZ},
\[
GQ(0)H
\in
Q(\alpha g+\alpha h)+O(\varepsilon)
=
Q(\alpha p)+O(\varepsilon).
\]
The number of primitive factors is at most
\[
\begin{aligned}
t_g+t_h+1
&\leq
2\cdot 45\cdot 9^{d-1}+1 \\
&\leq
45\cdot 9^d.
\end{aligned}
\]
The error degree is at most
\[
\begin{aligned}
e_g+e_h
&\leq
2\cdot 12\cdot 25^{d-1} \\
&\leq
12\cdot 25^d.
\end{aligned}
\]
No rescaling of \(\varepsilon\) occurs in this construction, and the
only new primitive label is \(0\). Hence the largest absolute Laurent
exponent is at most
\[
L_{d-1}\leq L_d.
\]

Suppose next that
\[
p=gh,
\]
where \(g\) and \(h\) are computed by formulas of depth at most
\(d-1\). Apply the induction hypothesis to the pairs
\[
(g,\alpha/2)
\qquad\text{and}\qquad
(h,1).
\]
This gives products
\[
G
=
Q((\alpha/2)g)
+
\sum_{r=1}^{e_g}\varepsilon^r B_r
\]
and
\[
H
=
Q(h)
+
\sum_{r=1}^{e_h}\varepsilon^r C_r,
\]
with
\[
t_g,t_h \leq 45\cdot 9^{d-1},
\qquad
e_g,e_h \leq 12\cdot 25^{d-1}.
\]

Replace \(\varepsilon\) by \(\varepsilon^3\) in every primitive label
of both products. Denote the resulting products by \(G^{[3]}\) and
\(H^{[3]}\). Then
\[
G^{[3]}
=
Q((\alpha/2)g)
+
\sum_{r=1}^{e_g}\varepsilon^{3r}B_r
\in
Q((\alpha/2)g)+O(\varepsilon^3)
\]
and
\[
H^{[3]}
=
Q(h)
+
\sum_{r=1}^{e_h}\varepsilon^{3r}C_r
\in
Q(h)+O(\varepsilon^3).
\]
Their error degrees are at most
\[
3e_g
\qquad\text{and}\qquad
3e_h,
\]
respectively. Moreover, every Laurent exponent occurring in their
primitive labels has absolute value at most
\[
3L_{d-1}.
\]

We now apply the multiplication gadget of
\cite[Lemma~3.4]{BIZ}. In the notation of that lemma, take
\[
f=\alpha g
\qquad\text{and}\qquad
g=h.
\]
The first input required by the gadget is an approximation to
\[
Q(f/2)
=
Q((\alpha g)/2)
=
Q((\alpha/2)g),
\]
which is exactly the matrix approximated by \(G^{[3]}\). The second
input is the matrix approximated by \(H^{[3]}\). The gadget therefore
produces
\[
Q(\ell_t)\cdots Q(\ell_1)
=
Q((\alpha g)h)+O(\varepsilon)
=
Q(\alpha p)+O(\varepsilon).
\]

The number of primitive factors is at most
\[
\begin{aligned}
4t_g+4t_h+37
&\leq
8\cdot 45\cdot 9^{d-1}+37 \\
&\leq
45\cdot 9^d.
\end{aligned}
\]
The error degree is at most
\[
4(3e_g)+4(3e_h)+12
=
12(e_g+e_h+1).
\]
Using the induction bounds, we obtain
\[
\begin{aligned}
12(e_g+e_h+1)
&\leq
12\bigl(24\cdot 25^{d-1}+1\bigr) \\
&=
288\cdot 25^{d-1}+12 \\
&\leq
300\cdot 25^{d-1} \\
&=
12\cdot 25^d.
\end{aligned}
\]
Here the penultimate inequality uses \(d\geq 1\), and therefore
\(25^{d-1}\geq 1\).

The fixed primitive labels inserted by the multiplication gadget have
Laurent exponents of absolute value at most \(2\). Therefore the
largest absolute Laurent exponent in the resulting product is at most
\[
\begin{aligned}
\max\{3L_{d-1},2\}
&=
\max\{2\cdot 3^d,2\} \\
&=
2\cdot 3^d \\
&=
L_d.
\end{aligned}
\]
This completes the induction.

It remains to verify the claimed syntactic-effectiveness bound. At the
leaves, all primitive labels are affine Laurent forms in the original
variables. The addition construction merely concatenates previously
constructed products and inserts the fixed form \(0\). The
multiplication construction only repeats previously constructed
labels, applies the rescaling
\[
\varepsilon\longmapsto\varepsilon^3,
\]
and inserts the fixed affine Laurent forms occurring in the gadgets of
\cite[Lemmas~3.2--3.4]{BIZ}. Hence every \(\ell_i\) is an affine
Laurent form in the original variables, and every such form has
Laurent support bounded by an absolute constant.

Every exponent occurring in a coefficient of an \(\ell_i\) has
absolute value at most
\[
2\cdot 3^d,
\]
and hence has binary description length \(O(d)\). Moreover, a fan-in-two formula of depth \(d\) has at most \(2^d\)
leaf occurrences. After renumbering the variables that actually occur
in the formula, every variable index therefore has binary description
length \(O(d)\). Since there are at most
\[
45\cdot 9^d
\]
forms, their total syntactic description length is
\[
O\bigl((45\cdot 9^d)(d+1)\bigr),
\]
with field constants counted as atomic labels. In particular, this
length is polynomial in \(45\cdot 9^d\) and \(d\).

Finally, take
\[
\alpha=1
\qquad\text{and}\qquad
p=f.
\]
The induction gives affine Laurent forms
\[
\ell_1,\ldots,\ell_t
\]
such that
\[
Q(\ell_t)\cdots Q(\ell_1)
=
Q(f)+O(\varepsilon),
\]
where
\[
t\leq 45\cdot 9^d,
\qquad
E\leq 12\cdot 25^d,
\]
and every Laurent exponent occurring in a coefficient of one of the
forms has absolute value at most
\[
2\cdot 3^d.
\]

By the definition of \(O(\varepsilon)\), the error belongs to
\[
\varepsilon
\operatorname{Mat}_{2\times 2}(F[\varepsilon,x]).
\]
Since its \(\varepsilon\)-degree is at most \(E\), it can be written
as
\[
Q(\ell_t)\cdots Q(\ell_1)
=
Q(f)
+
\sum_{r=1}^{E}\varepsilon^r A_r,
\qquad
A_r\in \operatorname{Mat}_{2\times 2}(F[x]),
\]
after inserting zero coefficient matrices for any missing powers of
\(\varepsilon\). This proves all the assertions.
\end{proof}

\begin{theorem}[Effective continuant simulation]\label[theorem]{thm:effective-continuant-simulation}
Assume $\operatorname{char}\F\ne2$.  Every family in $\VPe$ is an effective
polynomial-error Laurent degeneration of the continuant family:
\[
(f_n)\in\VPe
\quad\Longrightarrow\quad
(f_n)\DegRed(K_n),
\]
with polynomially bounded reindexing of the continuant.
\end{theorem}

\begin{proof}
Let $s(n)$ be a polynomial bound on formula size.  Brent's balancing theorem
\cite{Brent} supplies formulas of size polynomial in $s(n)$ and depth
$d(n)=O(\log s(n))=O(\log n)$.  Apply \Cref{prop:quant}.  Its bounds become polynomial in $n$:
\[
45\cdot9^{d(n)}=n^{O(1)},\qquad
12\cdot25^{d(n)}=n^{O(1)},\qquad
2\cdot3^{d(n)}=n^{O(1)}.
\]
Taking the $(1,1)$ entry and using \eqref{eq:continuant-matrix} gives
\[
K_{t(n)}(\ell_1,\ldots,\ell_{t(n)})
=f_n+\sum_{r=1}^{E(n)}\eps^rh_{n,r}.
\]
All effectiveness requirements in \cref{def:degeneration} are therefore
satisfied.
\end{proof}

\subsection{The constant-depth hierarchy and depth preservation}

The continuant simulation alone does not yield a lower bound: to derive a contradiction, 
we also need explicit polynomial families that can be computed at some fixed depth but cannot be computed at a slightly smaller fixed depth. 
Such families are provided by the adjacent-depth hierarchy theorem of Limaye, Srinivasan, and Tavenas. 
This theorem expresses the fact that constant-depth arithmetic computation does not collapse to a single depth: 
for every fixed depth bound, allowing one additional layer gives strictly greater polynomial-size computational power. 
The circuit model of Limaye, Srinivasan, and Tavenas allows a sum gate with inputs \(u_1,\ldots,u_r\) to compute an arbitrary \(F\)-linear 
combination \[ \alpha_1u_1+\cdots+\alpha_ru_r, \qquad \alpha_1,\ldots,\alpha_r\in F, \] with the coefficients represented as labels on the incoming edges; 
see \cite[p.~1, Footnote~1]{LST}. 
Since the model used here has only ordinary \(+\) and \(\times\) gates, we record the elementary translation between the two conventions.

\begin{remark}[Comparison of the two circuit models]
There are two minor differences between the circuit model of
Limaye, Srinivasan, and Tavenas and the ordinary arithmetic-circuit
model used in this paper.

First, a sum gate in the model of Limaye, Srinivasan, and Tavenas may
compute an arbitrary linear combination
\[
\alpha_1u_1+\cdots+\alpha_tu_t,
\qquad
\alpha_1,\ldots,\alpha_t\in F,
\]
where the coefficients are labels on the incoming edges. In the model
used here, addition gates compute ordinary sums, while multiplication
by a field constant must be performed explicitly. A linear-combination
gate can therefore be replaced by \(t\) parallel multiplication gates
computing
\[
\alpha_1u_1,\ldots,\alpha_tu_t,
\]
followed by one ordinary unbounded-fan-in addition gate. Thus each
linear-combination gate is replaced by at most two arithmetic layers.
Product gates require no modification.

Second, the two papers use slightly different size conventions.
Limaye, Srinivasan, and Tavenas measure circuit size by the number of
gates, whereas in this paper size is measured by the number of wires.
A circuit with \(s\) gates has at most \(s^2\) wires. Consequently, the
translation described above introduces at most polynomially many
ordinary gates and wires: the number of scalar-multiplication gates is
at most the number of incoming wires of the original
linear-combination gates, and hence is \(O(s^2)\). In particular, a
polynomial-size circuit in their convention becomes a polynomial-size
circuit in ours.

Conversely, an ordinary addition gate is a linear-combination gate
whose edge coefficients are all equal to \(1\), and an ordinary
multiplication gate is already allowed in their model. Moreover, an
ordinary circuit with polynomially many wires has polynomially many
relevant gates after unused gates are removed. Hence a
polynomial-size ordinary circuit of depth \(d\) is also a
polynomial-size depth-\(d\) circuit in the model of Limaye,
Srinivasan, and Tavenas.

It follows that the forward translation preserves polynomial size
and increases depth by at most a factor of \(2\), while the reverse
translation preserves depth. These are the only model-conversion
facts needed below.
\end{remark}

\begin{theorem}[Constant-depth hierarchy]\label[theorem]{thm:constant-depth-hierarchy}
Let $\F$ have characteristic zero.  For every fixed integer $D\geq2$, there is
an explicit polynomial family $H^{(D)}=(H_n^{(D)})_{n\geq1}$ such that:
\begin{enumerate}[label=(\roman*)]
\item $H^{(D)}$ has polynomial-size division-free $+$, $\times$ circuits of
      depth at most $2D$;
\item $H^{(D)}$ has no polynomial-size $+$, $\times$ circuits of depth at most
      $D-1$.
\end{enumerate}
\end{theorem}

\begin{proof}
By \cite[Theorem~5]{LST}, for every fixed $D\geq2$ and growing size parameter
$s$, there is an explicit set-multilinear polynomial $Q_{D,s}$ computed by a
depth-$D$ circuit of size at most $s$ in their model, whereas every
depth-$(D-1)$ circuit in that model computing $Q_{D,s}$ has size
$s^{\omega(1)}$.  Their footnote following the theorem specifies that this is
ordinary circuit depth, not merely product-depth.

Define
\[
H_n^{(D)}:=Q_{D,n}
\]
for all sufficiently large $n$, and define the finitely many remaining
members arbitrarily.  The circuit supplied by~\cite[Theorem~5]{LST}
has at most $n$ gates in their size convention.  Consequently, it has
at most $n^2$ wires, and the number of variables that actually occur
in $Q_{D,n}$ is also polynomially bounded in $n$.  Thus
$\bigl(H_n^{(D)}\bigr)_{n\geq1}$ is a polynomially indexed family.

Consider a linear-combination gate
\[
v=\alpha_1u_1+\cdots+\alpha_tu_t
\]
in the circuit of Limaye, Srinivasan, and Tavenas.  Replace it by
$t$ parallel scalar-multiplication gates computing
\[
\alpha_1u_1,\ldots,\alpha_tu_t,
\]
followed by one unbounded-fan-in ordinary addition gate.  Product gates
are left unchanged.

Along every input--output path, each original linear-combination gate
is thereby replaced by at most two ordinary arithmetic layers, while
each product gate remains one layer.  Hence an original circuit of
depth at most $D$ becomes an ordinary $+,\times$ circuit of depth at
most $2D$.

The number of newly introduced scalar-multiplication gates is at most
the number of incoming wires to linear-combination gates.  Since a
circuit with at most $n$ gates has at most $n^2$ wires, the translated
circuit has polynomially many gates and wires.  Thus
$H^{(D)}$ has polynomial-size ordinary $+,\times$ circuits of depth at
most $2D$.

Conversely, an ordinary addition gate is a linear-combination gate
whose incoming coefficients are all equal to $1$, and an ordinary
product gate is already permitted in the model of
Limaye, Srinivasan, and Tavenas.  Therefore any polynomial-size
ordinary $+,\times$ circuit of depth at most $D-1$ computing
$H_n^{(D)}$ would also be a polynomial-size depth-$(D-1)$ circuit in
their model.  This contradicts the $n^{\omega(1)}$ lower bound of
\cite[Theorem~5]{LST}.
\end{proof}

\begin{lemma}[Unfolding a constant-depth circuit]\label[lemma]{lem:constant-depth-unfolding}
A division-free arithmetic circuit of size $s$ and depth $D$ can be
converted into a fan-in-two division-free arithmetic formula of size
\[
O\bigl((D+1)s^{D+1}\bigr).
\]
In particular, its formula size is $s^{O(D)}$. Consequently, if a polynomial family \((f_n)\) is computed by
division-free arithmetic circuits of polynomial size and fixed depth,
then \((f_n)\in\mathrm{VP}_e\).
\end{lemma}

\begin{proof}
We first unfold all sharing in the circuit. Starting at the output gate,
replace each occurrence of a gate by a fresh copy together with fresh
copies of all subcircuits feeding it. Equivalently, a gate is copied once
for every directed path from that gate to the output.

Because the original circuit has $s$ wires, every gate has fan-in at
most $s$. Consequently, the unfolded computation has at most $s^k$
gate occurrences at distance $k$ from the output. Since the circuit has
depth $D$, the total number of gate occurrences in the unfolded
computation is at most
\[
1+s+s^2+\cdots+s^D
\leq
(D+1)s^D.
\]
The resulting directed graph is a formula: every non-output gate
occurrence has fan-out one, because a separate copy was made for each use
of its value.

The gates in the unfolded formula may still have unbounded fan-in.
Replace every addition or multiplication gate of fan-in \(r\) by a
fan-in-two binary tree containing \(r-1\) gates of the same type. Since
the original circuit has \(s\) wires, every gate has fan-in at most \(s\).
Therefore, each gate occurrence in the unfolded formula is replaced by
at most \(s-1\) binary gates. As the unfolded formula has at most
\[
1+s+s^2+\cdots+s^D \leq (D+1)s^D
\]
gate occurrences, the resulting fan-in-two formula has size at most
\[
O\bigl(s(1+s+s^2+\cdots+s^D)\bigr)
=
O\bigl((D+1)s^{D+1}\bigr)
=
s^{O(D)}.
\]

All transformations preserve the polynomial computed at the output and
introduce neither division gates nor new field operations.

Now consider a polynomial family computed by circuits of size
$n^{O(1)}$ and fixed depth $D$. The preceding construction gives
fan-in-two division-free formulas of size
\[
\bigl(n^{O(1)}\bigr)^{O(D)}=n^{O(1)}.
\]
Hence the family belongs to $\mathrm{VP}_e$.
\end{proof}

\begin{lemma}[Effective degeneration preserves fixed depth]\label[lemma]{lem:degeneration-preserves-depth}
Suppose $(g_m)$ has polynomial-size division-free circuits of depth $d$, and
$(f_n) \preceq_{DEG} (g_m)$.
Then $(f_n)$ has polynomial-size circuits of depth at most $d+4$.
\end{lemma}

\begin{proof}
Write
\[
G_n(\eps,x)=g_{m(n)}(\ell_1(\eps,x),\ldots,\ell_{t(n)}(\eps,x))
=f_n(x)+\sum_{r=1}^{E(n)}\eps^rh_{n,r}(x).
\]
Choose distinct nonzero elements
\[
\alpha_0,\ldots,\alpha_{E(n)}\in F.
\]
Lagrange interpolation supplies constants
\(\lambda_0,\ldots,\lambda_{E(n)}\in F\) such that, for every
polynomial
\[
P(\varepsilon)\in F[x_1,\ldots,x_m][\varepsilon]
\]
of \(\varepsilon\)-degree at most \(E(n)\),
\[
[\varepsilon^0]P(\varepsilon)
=
\sum_{j=0}^{E(n)}\lambda_jP(\alpha_j).
\]
Applying this identity to
\[
P(\varepsilon)=G_n(\varepsilon,x)
\]
and using
\(
[\varepsilon^0]G_n(\varepsilon,x)=f_n(x)
\),
we obtain
\[
f_n(x)
=
\sum_{j=0}^{E(n)}
\lambda_jG_n(\alpha_j,x).
\]
Each specialized affine form is computed in two layers; the depth-$d$ circuit
for $g_{m(n)}$ follows; one scalar-multiplication layer and one final addition
layer perform interpolation.  The depth is at most $d+4$, and effectiveness
ensures polynomial size.
\end{proof}

\begin{theorem}[Continuant lower bound]\label[theorem]{thm:continuant-lower-bound}
Over every field of characteristic zero, the continuant family is not in
$\AC^0_{\F}$.
\end{theorem}

\begin{proof}
Suppose the continuant had polynomial-size circuits of fixed depth $d$.
By \cref{thm:effective-continuant-simulation,lem:degeneration-preserves-depth},
every family in $\VPe$ would then have polynomial-size circuits of depth at
most $d+4$.

Choose the hierarchy family $H^{(d+5)}$ from
\cref{thm:constant-depth-hierarchy}.  It has circuits of fixed depth
$2(d+5)$ and hence, by \cref{lem:constant-depth-unfolding}, belongs to $\VPe$.
The assumed continuant upper bound would give depth-$(d+4)$ circuits for this
family, contradicting the depth-$(d+4)$ lower bound supplied by
\cref{thm:constant-depth-hierarchy} with $D=d+5$.
\end{proof}

\section{An explicit continuant--Hankel degeneration}\label{sec:degeneration}

We now give the first new reduction.  Separating the weight computation from
the combinatorial and sign analysis makes the mechanism of the degeneration
explicit.

\begin{theorem}[Path matchings inside a Hankel determinant]
\label{thm:path-matching-hankel}
Let $n\geq 1$ and let
\[
P_n=1-2-\cdots-n
\]
be the path introduced in Definition~3.2, and define
\[
M_{P_n}(x_1,\ldots,x_n)
 :=
 \sum_{T\in\operatorname{Match}(P_n)}
 \prod_{j\notin V(T)}x_j.
\]
Set
\[
N=n+1,\qquad
M=2n^2,\qquad
E_n=NM-2\sum_{i=0}^{n}i^2
   =\frac{n(n+1)(4n-1)}{3},
\]
and
\[
\sigma_n=(-1)^{\lfloor n/2\rfloor}.
\]
For $0\leq i<n$, put
\[
A_{2i}=(-1)^i x_{i+1}\varepsilon^{M-2i^2},
\qquad
A_{2i+1}
   =\varepsilon^{M-i^2-(i+1)^2},
\]
and put
\[
A_{2n}=\sigma_n\varepsilon^{-E_n}.
\]
Then
\[
[\varepsilon^0]\,
\Delta_{n+1}(A_0,\ldots,A_{2n})
 =
M_{P_n}(x_1,\ldots,x_n)
 =
K_n(x_1,\ldots,x_n).
\]
Moreover, the substitution defines an effective polynomial-error
Laurent degeneration
\[
K_n\preceq_{\mathrm{DEG}}\Delta_{n+1}.
\]
\end{theorem}

\begin{lemma}[Weight of a determinant term]\label[lemma]{lem:hankel-term-weight}
For each permutation $\pi\in S_{n+1}$, let
\[
T_{\pi}
:=
\operatorname{sgn}(\pi)
\prod_{i=0}^{n} A_{i+\pi(i)}
\]
denote the summand indexed by $\pi$ in the Leibniz expansion
\[
\Delta_{n+1}(A_0,\ldots,A_{2n})
=
\sum_{\pi\in S_{n+1}}
\operatorname{sgn}(\pi)
\prod_{i=0}^{n} A_{i+\pi(i)}.
\]
Define
\[
P(\pi)
:=
\sum_{i=0}^{n}
\left\lfloor
\frac{(i-\pi(i))^2}{2}
\right\rfloor.
\]
Then $T_\pi$ is a monomial in $\varepsilon$ whose
$\varepsilon$-exponent is
\begin{equation}
\operatorname{exp}_{\varepsilon}(T_\pi)
=
\begin{cases}
P(\pi), & \pi(n)=n,\\[4pt]
E_n+P(\pi), & \pi(n)\neq n.
\end{cases}
\label{eq:term-exponent}
\end{equation}

In particular, every term in the Leibniz expansion has
nonnegative $\varepsilon$-exponent.
\end{lemma}

\begin{proof}
For $0\leq s\leq 2n$, define the provisional weight
\[
\widetilde{w}_s
:=
M-\left\lceil\frac{s^2}{2}\right\rceil.
\]
For every $s<2n$, this is exactly the exponent of $\varepsilon$
in $A_s$. Indeed, for $0\leq r<n$,
\[
\widetilde{w}_{2r}
=
M-\left\lceil\frac{(2r)^2}{2}\right\rceil
=
M-2r^2,
\]
and
\[
\widetilde{w}_{2r+1}
=
M-\left\lceil\frac{(2r+1)^2}{2}\right\rceil
=
M-r^2-(r+1)^2.
\]
For $s=2n$, we have
\[
\widetilde{w}_{2n}
=
M-\left\lceil\frac{(2n)^2}{2}\right\rceil
=
2n^2-2n^2
=
0,
\]
whereas the actual exponent of $\varepsilon$ in
\[
A_{2n}=\sigma_n\varepsilon^{-E_n}
\]
is $-E_n$.

Ignoring this final correction for the moment, the provisional
$\varepsilon$-exponent of the term $T_\pi$ is
\[
\begin{aligned}
\widetilde{e}(\pi)
&=
\sum_{i=0}^{n}
\widetilde{w}_{i+\pi(i)}\\
&=
\sum_{i=0}^{n}
\left(
M-
\left\lceil
\frac{(i+\pi(i))^2}{2}
\right\rceil
\right).
\end{aligned}
\]
For arbitrary integers $i$ and $j$, we have
\[
\left\lceil\frac{(i+j)^2}{2}\right\rceil
=
i^2+j^2-
\left\lfloor\frac{(i-j)^2}{2}\right\rfloor.
\]
Applying this identity with $j=\pi(i)$ gives
\[
\begin{aligned}
\widetilde{e}(\pi)
&=
\sum_{i=0}^{n}
\left(
M-i^2-\pi(i)^2
+
\left\lfloor
\frac{(i-\pi(i))^2}{2}
\right\rfloor
\right)\\
&=
(n+1)M
-
\sum_{i=0}^{n}i^2
-
\sum_{i=0}^{n}\pi(i)^2
+
P(\pi).
\end{aligned}
\]
Since $\pi$ is a permutation of $\{0,\ldots,n\}$,
\[
\sum_{i=0}^{n}\pi(i)^2
=
\sum_{i=0}^{n}i^2.
\]
Hence, by the definition
\[
E_n
=
(n+1)M-2\sum_{i=0}^{n}i^2,
\]
we obtain
\[
\widetilde{e}(\pi)=E_n+P(\pi).
\]

It remains to account for the actual exponent of $A_{2n}$.
The factor $A_{2n}$ occurs in $T_\pi$ precisely when
\[
i+\pi(i)=2n
\]
for some $i\in\{0,\ldots,n\}$. Since both $i$ and $\pi(i)$ lie
between $0$ and $n$, this equality holds if and only if
\[
i=n
\qquad\text{and}\qquad
\pi(n)=n.
\]

If $\pi(n)\neq n$, the term $T_\pi$ does not contain $A_{2n}$,
so no correction is required, and
\[
\operatorname{exp}_{\varepsilon}(T_\pi)
=
E_n+P(\pi).
\]

If $\pi(n)=n$, the term contains $A_{2n}$. Its provisional
exponent was taken to be $0$, whereas its actual exponent is
$-E_n$. We must therefore subtract $E_n$, obtaining
\[
\operatorname{exp}_{\varepsilon}(T_\pi)
=
\widetilde{e}(\pi)-E_n
=
P(\pi).
\]

This proves the asserted formula. Finally,
\[
P(\pi)\geq 0
\]
for every $\pi$, and
\[
E_n=\frac{n(n+1)(4n-1)}{3}>0
\]
for $n\geq 1$. Thus every determinant term has nonnegative
$\varepsilon$-exponent.
\end{proof}

\begin{lemma}[Zero-weight permutations and their signs]\label[lemma]{lem:zero-weight-permutations}
For each permutation $\pi\in S_{n+1}$, let
\[
T_\pi(\varepsilon,x)
:=
\operatorname{sgn}(\pi)
\prod_{i=0}^{n} A_{i+\pi(i)}
\]
be the summand indexed by $\pi$ in the Leibniz expansion
\[
\Delta_{n+1}(A_0,\ldots,A_{2n})
=
\sum_{\pi\in S_{n+1}}
T_\pi(\varepsilon,x).
\]

For a matching $T\in\Match(P_n)$, define
$\pi_T\in S_{n+1}$ as follows:
\[
\pi_T(n)=n;
\]
for every edge $\{j,j+1\}\in T$, set
\[
\pi_T(j-1)=j,
\qquad
\pi_T(j)=j-1;
\]
and fix every remaining index in $\{0,\ldots,n-1\}$.

Then $T_\pi(\varepsilon,x)$ has $\varepsilon$-exponent zero if
and only if
\[
\pi=\pi_T
\]
for a unique matching $T\in\Match(P_n)$. Moreover, for that
matching,
\[
T_{\pi_T}(\varepsilon,x)
=
\prod_{j\notin V(T)}x_j.
\]
In particular, every zero-$\varepsilon$-exponent summand occurs
with positive sign.
\end{lemma}

\begin{proof}
By Lemma~\ref{lem:hankel-term-weight}, the
$\varepsilon$-exponent of
\[
T_\pi(\varepsilon,x)
=
\operatorname{sgn}(\pi)
\prod_{i=0}^{n}A_{i+\pi(i)}
\]
is
\[
\begin{cases}
P(\pi), & \pi(n)=n,\\[4pt]
E_n+P(\pi), & \pi(n)\neq n,
\end{cases}
\]
where
\[
P(\pi)
=
\sum_{i=0}^{n}
\left\lfloor
\frac{(i-\pi(i))^2}{2}
\right\rfloor.
\]
Since $E_n>0$ and every summand in $P(\pi)$ is nonnegative,
the $\varepsilon$-exponent is zero if and only if
\[
\pi(n)=n
\qquad\text{and}\qquad
P(\pi)=0.
\]

For an integer $d$,
\[
\left\lfloor\frac{d^2}{2}\right\rfloor=0
\quad\Longleftrightarrow\quad
|d|\leq 1.
\]
Consequently,
\[
P(\pi)=0
\quad\Longleftrightarrow\quad
|i-\pi(i)|\leq 1
\quad\text{for every }0\leq i\leq n.
\]
Since $\pi(n)=n$, it remains to understand the restriction of
$\pi$ to $\{0,\ldots,n-1\}$.

We claim that every permutation of $\{0,\ldots,n-1\}$
satisfying
\[
|i-\pi(i)|\leq 1
\]
is a product of disjoint adjacent transpositions.

Read the indices from left to right, and let $i$ be the smallest
index not already fixed or paired. If $\pi(i)=i$, then $i$ is
fixed. Otherwise, the displacement condition gives
\[
\pi(i)\in\{i-1,i+1\}.
\]
The value $i-1$ is impossible, because all smaller indices have
already been fixed or paired among themselves. Hence
\[
\pi(i)=i+1.
\]
Because $\pi$ is bijective, the value $i$ must be the image of
some index. The displacement condition shows that the only
possible remaining preimage is $i+1$, and therefore
\[
\pi(i+1)=i.
\]
Thus $i$ and $i+1$ form the adjacent transposition
\[
(i,i+1).
\]
Continuing from left to right proves that the restriction of
$\pi$ is a product of disjoint adjacent transpositions.

Under the index shift
\[
\{0,\ldots,n-1\}\longrightarrow\{1,\ldots,n\},
\qquad
i\longmapsto i+1,
\]
the transposition $(j-1,j)$ corresponds to the edge
\[
\{j,j+1\}
\]
of the path
\[
P_n=1-2-\cdots-n.
\]
Disjoint adjacent transpositions therefore correspond
bijectively to matchings $T\in\Match(P_n)$. The corresponding
permutation is precisely the explicitly defined permutation
$\pi_T$.

It remains to compute the summand
\[
T_{\pi_T}(\varepsilon,x)
=
\operatorname{sgn}(\pi_T)
\prod_{i=0}^{n}A_{i+\pi_T(i)}.
\]

Suppose first that a vertex $j\in\{1,\ldots,n\}$ is unmatched
in $T$. Then the index $i=j-1$ is fixed by $\pi_T$, so the
factor selected from row $i$ is
\[
A_{i+\pi_T(i)}
=
A_{2i}
=
A_{2(j-1)}.
\]
By the definition of the substitution,
\[
A_{2(j-1)}
=
(-1)^{j-1}x_j
\varepsilon^{M-2(j-1)^2}.
\]
Thus every unmatched vertex $j$ contributes the factor
\[
(-1)^{j-1}x_j
\]
to the coefficient of $T_{\pi_T}$.

Now suppose that $\{j,j+1\}\in T$. The associated
transposition is $(j-1,j)$, so
\[
\pi_T(j-1)=j,
\qquad
\pi_T(j)=j-1.
\]
The two selected factors are therefore
\[
A_{(j-1)+\pi_T(j-1)}
=
A_{2j-1}
\]
and
\[
A_{j+\pi_T(j)}
=
A_{2j-1}.
\]
Hence the edge $\{j,j+1\}$ contributes
\[
A_{2j-1}^2.
\]
Since
\[
A_{2j-1}
=
\varepsilon^{M-(j-1)^2-j^2},
\]
this contribution has coefficient $1$ and contains no
$x$-variable. The transposition itself contributes one factor
of $-1$ to $\operatorname{sgn}(\pi_T)$.

Finally, $\pi_T(n)=n$, so the last index contributes
\[
A_{n+\pi_T(n)}
=
A_{2n}
=
\sigma_n\varepsilon^{-E_n}.
\]

Since $\pi_T$ consists of exactly $|T|$ transpositions,
\[
\operatorname{sgn}(\pi_T)=(-1)^{|T|}.
\]
Therefore, after collecting all factors, we obtain
\[
T_{\pi_T}(\varepsilon,x)
=
\sigma_n
(-1)^{|T|+\sum_{j\notin V(T)}(j-1)}
\left(\prod_{j\notin V(T)}x_j\right)
\varepsilon^{\,0},
\]
where the total exponent of $\varepsilon$ is zero by the first
part of the proof. Hence
\[
T_{\pi_T}(\varepsilon,x)
=
\sigma_n
(-1)^{|T|+\sum_{j\notin V(T)}(j-1)}
\prod_{j\notin V(T)}x_j.
\]

We now verify that the displayed sign is always positive. Put
\[
S
:=
\sum_{j=1}^{n}(j-1)
=
\frac{n(n-1)}{2}.
\]
For each matched edge $\{j,j+1\}\in T$, the two corresponding
terms removed from this sum are
\[
(j-1)+j=2j-1,
\]
which is odd. Thus
\[
S
=
\sum_{j\notin V(T)}(j-1)
+
\sum_{\{j,j+1\}\in T}(2j-1).
\]
Reducing modulo $2$ gives
\[
S
\equiv
\sum_{j\notin V(T)}(j-1)+|T|
\pmod 2.
\]
Consequently,
\[
(-1)^{|T|+\sum_{j\notin V(T)}(j-1)}
=
(-1)^S.
\]

By the definition
\[
\sigma_n=(-1)^{\lfloor n/2\rfloor},
\]
and since
\[
\frac{n(n-1)}{2}
\equiv
\left\lfloor\frac n2\right\rfloor
\pmod 2,
\]
we have
\[
(-1)^S=\sigma_n.
\]
The total sign is therefore
\[
\sigma_n(-1)^S
=
\sigma_n^2
=
1.
\]
It follows that
\[
T_{\pi_T}(\varepsilon,x)
=
\prod_{j\notin V(T)}x_j,
\]
as claimed.
\end{proof}

\begin{proof}[Proof of \Cref{thm:path-matching-hankel}]
Expand
\[
\Delta_{n+1}(A_0,\ldots,A_{2n})
=\sum_{\pi\in S_{n+1}}\operatorname{sgn}(\pi)
  \prod_{i=0}^{n}A_{i+\pi(i)}.
\]
By \cref{lem:hankel-term-weight,lem:zero-weight-permutations}, the constant
coefficient is exactly
\[
\sum_{T\in\operatorname{Match}(P_n)}
\prod_{j\notin V(T)}x_j
 =
M_{P_n}(x_1,\ldots,x_n).
\]
Partitioning the matchings according to whether the last vertex $n$
is unmatched or is matched to $n-1$ gives
\[
M_{P_n}
 =
x_n M_{P_{n-1}}+M_{P_{n-2}},
\qquad
M_{P_0}=1,\qquad M_{P_1}=x_1.
\]
Thus
\[
M_{P_n}=K_n.
\]

It remains only to record effectiveness.  The largest absolute exponent in the
substitution is at most $\max\{E_n,M\}=O(n^3)$.  Moreover,
\[
P(\pi)\leq\frac{(n+1)n^2}{2},
\]
so \eqref{eq:term-exponent} shows that the degree in $\eps$ after substitution
is at most
\[
E_n+\frac{(n+1)n^2}{2}=O(n^3).
\]
There are $2n+1$ Laurent forms, each of constant Laurent support and
polynomial description length.  Hence all requirements of
\cref{def:degeneration} are satisfied.
\end{proof}

\begin{remark}[The degeneration in a small case]
The case \(n=3\) illustrates explicitly how the construction in
Theorem~4.1 extracts the matching terms from the Hankel determinant.
Here
\[
M=18,
\qquad
E_3=44,
\qquad
\sigma_3=-1,
\]
and the prescribed substitution is
\[
\begin{aligned}
A_0&=x_1\varepsilon^{18},
&
A_1&=\varepsilon^{17},
&
A_2&=-x_2\varepsilon^{16},
\\
A_3&=\varepsilon^{13},
&
A_4&=x_3\varepsilon^{10},
&
A_5&=\varepsilon^5,
&
A_6&=-\varepsilon^{-44}.
\end{aligned}
\]
A direct expansion gives
\[
[\varepsilon^0]\,
\Delta_4(A_0,\ldots,A_6)
=
x_1x_2x_3+x_1+x_3
=
K_3(x_1,x_2,x_3).
\]
The three surviving monomials correspond respectively to the empty
matching of the path \(1-2-3\), the matching containing the edge
\(\{2,3\}\), and the matching containing the edge \(\{1,2\}\).
Thus this example already displays the general mechanism: the
constant-weight determinant terms are precisely those indexed by
matchings of the path.
\end{remark}

\begin{proof}[Proof of \Cref{thm:hankel-lb}]
If $(\Delta_N)$ had polynomial-size circuits of fixed depth $d$, then
\cref{thm:path-matching-hankel,lem:degeneration-preserves-depth} would give
polynomial-size fixed-depth circuits for the continuant.  This contradicts
\cref{thm:continuant-lower-bound}.
\end{proof}

\section{Fixed-bound principal subresultants and normal Euclidean sequences}\label{sec:subresultants}

The purpose of this section is to recover a principal subresultant
coefficient from fixed coordinates of a normal Euclidean remainder
sequence.  \Cref{sec:hankel-psc} identifies the relevant principal
subresultant with a consecutive Hankel determinant, and \Cref{sec:euclid-to-hankel} combines the two facts.  The key formula proved below is
\[
    \operatorname{psc}_{m-i}(R_0,R_1)
    =
    \operatorname{lc}(R_i)
    \prod_{k=1}^{i-1}\operatorname{lc}(R_k)^2.
\]

For classical background on Euclidean remainder sequences and subresultants, see \cite{BrownTraub,vzGG}.  All identities used below are nevertheless proved directly.

A fixed-bound definition is convenient because the coefficients will later be specialized to loci on which the actual degree may drop.
It also makes every principal subresultant coefficient a polynomial in the displayed coefficient variables.

Let
\[
F=f_0x^m+f_1x^{m-1}+\cdots+f_m,
\qquad
G=g_0x^r+g_1x^{r-1}+\cdots+g_r,
\qquad m>r,
\]
where $m$ and $r$ are fixed degree bounds; we do not initially require $f_0$ or $g_0$ to be nonzero.

\begin{definition}[Fixed-bound principal subresultant coefficient]\label[definition]{def:psc}
For $0\le j\le r$, let $M_j^{m,r}(F,G)$ be the square matrix whose rows are the coefficient vectors, in descending powers from $x^{m+r-j-1}$ to $x^j$, of
\[
x^{r-j-1}F,\ldots,F,
\qquad
x^{m-j-1}G,\ldots,G.
\]
The first list is empty when $j=r$.
There are
\[
(r-j)+(m-j)=m+r-2j
\]
rows and columns.
Define
\[
\psc_j^{m,r}(F,G):=\det M_j^{m,r}(F,G).
\]
When the actual degrees of $F$ and $G$ are $m$ and $r$, respectively,
the determinant above is the principal subresultant coefficient in the
row-and-column convention fixed by Definition~5.1.  Some references use
an ordering convention that differs from ours by a predictable sign.
Throughout this paper, the notation
\[
\operatorname{psc}^{m,r}_j(F,G)
\]
means exactly the determinant specified in Definition~5.1.  When the
degree bounds are clear from context, we suppress the superscript
$(m,r)$.
\end{definition}

When $m=r+1$ and $j=r$, the matrix is $1\times1$ and
\[
\psc_r^{r+1,r}(F,G)=g_0.
\]
If $\deg G=r$, this is $\lc(G)$.

\subsection{One Euclidean step}

\begin{lemma}[Subresultant descent for consecutive degrees]\label[lemma]{lem:psc-descent}
Let $p\geq2$, let
\[
\deg F=p,\qquad \deg G=p-1,
\]
and set $g=\lc(G)$.  Write one ordinary division step as
\[
F=QG+H,
\qquad \deg Q=1,
\qquad \deg H\leq p-2.
\]
Then
\[
\psc_{p-2}(F,G)=g^2[x^{p-2}]H.
\]
If $\deg H=p-2$, then for every $0\le j\le p-2$,
\[
\psc_j(F,G)=g^2\psc_j(G,H).
\]
\end{lemma}

\begin{proof}
Fix $j\le p-2$.
The rows of $M_j(F,G)$ are
\[
x^{p-j-2}F,\ldots,F,
\qquad
x^{p-j-1}G,\ldots,G,
\]
and the columns are indexed by descending degrees
\[
2p-j-2,\ 2p-j-3,\ldots,j.
\]
For each first-block row $x^kF$, subtract the coefficient vector of $x^kQG$.
Since $\deg Q=1$, every shifted copy of $G$ needed for this row operation occurs in the second block.
The determinant is unchanged and the first block becomes
\[
x^{p-j-2}H,\ldots,H.
\]

Move the $G$ block before the $H$ block.
The block lengths are $p-j$ and $p-j-1$, so the sign is
\[
(-1)^{(p-j)(p-j-1)}=1.
\]
The first two $G$ rows, corresponding to $x^{p-j-1}G$ and $x^{p-j-2}G$, form a triangular pivot block in the first two columns with diagonal entries $g,g$.
Every other row has degree at most $2p-j-4$ and hence has zero in these two columns.
Expansion along the first two rows and columns contributes $g^2$ and no additional cofactor sign.

If $\deg H=p-2$, the remaining rows are
\[
x^{p-j-3}G,\ldots,G,
\qquad
x^{p-j-2}H,\ldots,H,
\]
and the remaining columns run from degree $2p-j-4$ down to $j$.
This is exactly $M_j(G,H)$, proving
\[
\psc_j(F,G)=g^2\psc_j(G,H).
\]
For $j=p-2$, the remaining matrix is $1\times1$ with entry $[x^{p-2}]H$, whether or not this coefficient is zero.
\end{proof}

\begin{definition}[Normal Euclidean sequence]\label[definition]{def:normal}
A Euclidean remainder sequence
\[
R_0,R_1,\ldots,R_s,
\qquad R_{i-1}=Q_iR_i+R_{i+1},
\]
is \emph{normal through $R_s$} if
\[
\deg R_i=m-i
\qquad(0\le i\le s),
\]
where $m=\deg R_0$ and $\deg R_1=m-1$.
\end{definition}

\begin{proposition}[Normal scaling formula]\label[proposition]{prop:normal-scaling}
Suppose $R_0,R_1,\ldots,R_s$ is normal, and set $c_i=\lc(R_i)$.
Then for every $1\le i\le s$,
\[
\psc_{m-i}(R_0,R_1)
=c_i\prod_{k=1}^{i-1}c_k^2,
\]
where the product is empty and equal to $1$ when $i=1$.
\end{proposition}

\begin{proof}
For $i=1$, the pair $(R_0,R_1)$ has degrees $m,m-1$, and the top-index subresultant matrix is $1\times1$:
\[
\psc_{m-1}(R_0,R_1)=\lc(R_1)=c_1.
\]

Let $i\ge2$ and put $j=m-i$.
Apply \cref{lem:psc-descent} successively to
\[
(R_0,R_1),(R_1,R_2),\ldots,(R_{i-2},R_{i-1}).
\]
Normality guarantees consecutive degrees at every step, yielding
\[
\psc_j(R_0,R_1)
=\left(\prod_{k=1}^{i-1}c_k^2\right)
\psc_j(R_{i-1},R_i).
\]
The last pair has degrees $j+1,j$, so its top-index principal subresultant matrix is $1\times1$ and
\[
\psc_j(R_{i-1},R_i)=\lc(R_i)=c_i.
\]
\end{proof}

\begin{corollary}[Nonvanishing criterion for normality]
\label[corollary]{cor:normality}
Let
\[
R_{i-1}=Q_iR_i+R_{i+1},
\qquad
\deg R_{i+1}<\deg R_i,
\]
be the ordinary Euclidean remainder sequence of $R_0$ and
$R_1$, where
\[
\deg R_0=m,
\qquad
\deg R_1=m-1.
\]
Let $2\leq s\leq m$. If
\[
\operatorname{psc}_{m-2}(R_0,R_1),\,
\operatorname{psc}_{m-3}(R_0,R_1),\,
\ldots,\,
\operatorname{psc}_{m-s}(R_0,R_1)
\]
are all nonzero, then the Euclidean remainder sequence is normal
through $R_s$; that is,
\[
\deg R_i=m-i
\qquad
\text{for every }0\leq i\leq s.
\]
\end{corollary}

\begin{proof}
We prove by induction on $t$ that
\[
\deg R_i=m-i
\qquad
\text{for every }0\leq i\leq t.
\]

The assertion for $t=1$ follows immediately from the
assumptions
\[
\deg R_0=m,
\qquad
\deg R_1=m-1.
\]

Now let $1\leq t<s$, and assume our claim for $t$. We
will prove that
\[
\deg R_{t+1}=m-t-1,
\]
which would prove the induction step.

Set
\[
j:=m-t-1.
\]
Since $t+1\leq s$, the principal subresultant coefficient
\[
\operatorname{psc}_{j}(R_0,R_1)
=
\operatorname{psc}_{m-(t+1)}(R_0,R_1)
\]
is one of the coefficients assumed to be nonzero. Hence
\begin{equation}
\operatorname{psc}_{j}(R_0,R_1)\neq 0.
\label{eq:original-psc-nonzero}
\end{equation}

We first transport this nonvanishing condition down the already
normal part of the Euclidean sequence. For every
$1\leq k\leq t-1$, the induction hypothesis gives
\[
\deg R_{k-1}=m-k+1,
\qquad
\deg R_k=m-k,
\qquad
\deg R_{k+1}=m-k-1.
\]
Thus the pair $(R_{k-1},R_k)$ has consecutive degrees, and its
remainder $R_{k+1}$ has the expected degree. Moreover,
\[
j=m-t-1\leq m-k-1.
\]
Therefore the second assertion of
Lemma~\ref{lem:psc-descent} applies and gives
\[
\operatorname{psc}_{j}(R_{k-1},R_k)
=
\operatorname{lc}(R_k)^2
\operatorname{psc}_{j}(R_k,R_{k+1}).
\]

Iterating these identities for $k=1,\ldots,t-1$, we obtain
\begin{equation}
\operatorname{psc}_{j}(R_0,R_1)
=
\left(
\prod_{k=1}^{t-1}\operatorname{lc}(R_k)^2
\right)
\operatorname{psc}_{j}(R_{t-1},R_t).
\label{eq:iterated-psc-descent}
\end{equation}

When $t=1$, the product in \eqref{eq:iterated-psc-descent} is empty and is interpreted
as $1$, so \eqref{eq:iterated-psc-descent} is simply the identity
\[
\operatorname{psc}_{j}(R_0,R_1)
=
\operatorname{psc}_{j}(R_0,R_1).
\]

By \eqref{eq:original-psc-nonzero}, the left-hand side of \eqref{eq:iterated-psc-descent} is nonzero. Every
$R_k$ with $1\leq k\leq t$ is nonzero by the induction
hypothesis, so every leading coefficient
$\operatorname{lc}(R_k)$ is nonzero. It follows from \eqref{eq:iterated-psc-descent}
that
\begin{equation}
\operatorname{psc}_{j}(R_{t-1},R_t)\neq 0.
\label{eq:current-psc-nonzero}
\end{equation}

Now consider the next Euclidean division,
\[
R_{t-1}=Q_tR_t+R_{t+1}.
\]
By the induction hypothesis,
\[
\deg R_{t-1}=m-t+1,
\qquad
\deg R_t=m-t.
\]
If we put
\[
p:=m-t+1,
\]
then
\[
\deg R_{t-1}=p,
\qquad
\deg R_t=p-1,
\qquad
j=m-t-1=p-2.
\]
The first assertion of
Lemma~\ref{lem:psc-descent}, applied to the pair
$(R_{t-1},R_t)$, therefore gives
\begin{equation}
\operatorname{psc}_{j}(R_{t-1},R_t)
=
\operatorname{lc}(R_t)^2
[x^j]R_{t+1}.
\label{eq:psc-detects-leading-coefficient}
\end{equation}

The left-hand side of \eqref{eq:psc-detects-leading-coefficient} is nonzero by \eqref{eq:current-psc-nonzero}, and
$\operatorname{lc}(R_t)\neq0$. Hence
\[
[x^j]R_{t+1}\neq0.
\]
It follows that
\[
\deg R_{t+1}\geq j.
\]
On the other hand, since $R_{t+1}$ is the Euclidean remainder
upon division by $R_t$,
\[
\deg R_{t+1}<\deg R_t=m-t=j+1.
\]
Therefore
\[
\deg R_{t+1}\leq j.
\]
Combining the two inequalities yields
\[
\deg R_{t+1}=j=m-t-1.
\]

This proves the induction step. By induction,
\[
\deg R_i=m-i
\qquad
\text{for every }0\leq i\leq s,
\]
so the Euclidean sequence is normal through $R_s$.
\end{proof}

\section{A Hankel determinant as a middle principal subresultant}\label{sec:hankel-psc}

For $n\ge1$, introduce coefficients $b_1,\ldots,b_{2n}$ and set
\[
F_n(x)=x^{2n},
\qquad
G_b(x)=b_1x^{2n-1}+b_2x^{2n-2}+\cdots+b_{2n}.
\]
All principal subresultants in this section are understood with the fixed degree bounds $(2n,2n-1)$ from \cref{def:psc}, so the identities remain polynomial identities even after a specialization with $b_1=0$.

\begin{lemma}[Monomial--Hankel identity]\label[lemma]{lem:monomial-hankel}
For $1\le k\le n$,
\[
\psc_{2n-k}^{2n,2n-1}(F_n,G_b)
=(-1)^{\binom{k}{2}}
\det(b_{i+j+1})_{0\le i,j<k}.
\]
In particular,
\[
\psc_n^{2n,2n-1}(F_n,G_b)
=(-1)^{\binom{n}{2}}
\det(b_{i+j+1})_{0\le i,j<n}.
\]
\end{lemma}

\begin{proof}
Put $j=2n-k$.
The $F_n$ block of $M_j^{2n,2n-1}(F_n,G_b)$ contains the $k-1$ rows
\[
x^{k-2}F_n,\ldots,F_n.
\]
In the descending-degree columns of \cref{def:psc}, these are unit vectors in the first $k-1$ columns.
Clear the entries below these pivots and expand along them.

The remaining $k\times k$ matrix has rows, in order,
\[
x^{k-1}G_b,x^{k-2}G_b,\ldots,G_b.
\]
Index the rows by $r=0,\ldots,k-1$ and the remaining columns by $c=0,\ldots,k-1$.
A direct coefficient check gives the entry
\[
b_{k-r+c}.
\]
Reverse the order of the $k$ rows.
This uses $\binom{k}{2}$ transpositions and produces the Hankel matrix
\[
(b_{1+r+c})_{0\le r,c<k}.
\]
\end{proof}

Now specialize
\[
b_{r+1}=a_r\quad(0\le r\le2n-2),
\qquad b_{2n}=1,
\]
and write
\[
G_a(x)=a_0x^{2n-1}+a_1x^{2n-2}+\cdots+a_{2n-2}x+1.
\]
Then \cref{lem:monomial-hankel} gives the polynomial identity
\begin{equation}\label{eq:psc-hankel}
\psc_n^{2n,2n-1}(F_n,G_a)
=(-1)^{\binom n2}\Delta_n(a_0,\ldots,a_{2n-2}).
\end{equation}

\begin{lemma}[The normal locus is nonempty]\label[lemma]{lem:normal-locus}
There is a nonempty Zariski-open set
\[
U_n\subseteq F^{2n-1}
\]
such that, for every $a\in U_n$, one has $a_0\neq 0$ and the
ordinary Euclidean remainder sequence
\[
(R_0=F_n, R_1=G_a, R_2,\ldots)
\]
is normal through $R_n$; equivalently,
\[
\deg R_i=2n-i
\qquad
\text{for every }0\leq i\leq n.
\]
In particular,
\[
\deg R_n=n.
\]
\end{lemma}

\begin{proof}
For $1\le k\le n$, define the leading Hankel minor
\[
D_k(a):=\det(a_{i+j})_{0\le i,j<k}.
\]
By \cref{lem:monomial-hankel}, up to the fixed sign $(-1)^{\binom k2}$,
\[
D_k(a)=\psc_{2n-k}^{2n,2n-1}(F_n,G_a).
\]
In particular, $D_1=a_0$.
Define
\[
U_n:=\{a:D_1(a)D_2(a)\cdots D_n(a)\ne0\}.
\]
This set is Zariski open.
On $U_n$, $G_a$ has actual degree $2n-1$, and \cref{cor:normality}, applied to the nonzero subresultants for $k=2,\ldots,n$, proves normality through $R_n$.

It remains to prove nonemptiness.
Choose distinct positive rational numbers
\[
\lambda_1,\ldots,\lambda_n
\]
and set
\[
a_s=\sum_{t=1}^n\lambda_t^s
\qquad(0\le s\le2n-2).
\]
For $k\le n$, let
\[
V_k=(\lambda_t^i)_{\substack{0\le i<k\\1\le t\le n}}.
\]
Then
\[
(a_{i+j})_{0\le i,j<k}=V_kV_k^{\mathsf T}.
\]
The matrix $V_k$ has row rank $k$: a nonzero polynomial of degree less than $k$ cannot vanish at the $n\ge k$ distinct points $\lambda_t$.
Thus $V_kV_k^{\mathsf T}$ is positive definite over $\mathbb R$, so $D_k(a)>0$.
Each $D_k(a)$ is therefore a nonzero rational number.
Every characteristic-zero field contains a canonical copy of $\Q$, so the same rational specialization remains nonzero over $\F$.
\end{proof}

\subsection{Scaling the second input}

\begin{lemma}[Scaling of the relevant subresultants]\label[lemma]{lem:scaling}
For $1\le k\le n$ and $\lambda\in\F$,
\[
\psc_{2n-k}^{2n,2n-1}(F_n,\lambda G)
=\lambda^k\psc_{2n-k}^{2n,2n-1}(F_n,G).
\]
Consequently, when $\lambda\ne0$, the nonvanishing conditions defining normality through $R_n$ are invariant under replacing $G$ by $\lambda G$.
\end{lemma}

\begin{proof}
The matrix $M_{2n-k}^{2n,2n-1}(F_n,G)$ contains exactly
\[
2n-(2n-k)=k
\]
rows belonging to the second argument.
Scaling $G$ by $\lambda$ scales every one of these rows by $\lambda$.
The final assertion follows from \cref{cor:normality}.
\end{proof}

On the basic open set $D(a_0)$, define the monic rational normalization
\begin{equation}\label{eq:monic-normalization}
\widehat G_a:=a_0^{-1}G_a.
\end{equation}
This is monic of degree $2n-1$ wherever $a_0\ne0$.
The reduction below is performed only on $U_n\subseteq D(a_0)$, where the composed circuit is a formal piecewise rational circuit in the sense of \cref{def:piecewise}.
By \cref{lem:normal-locus,lem:scaling}, the Euclidean sequence of $(F_n,\widehat G_a)$ is normal through $R_n$ on $U_n$.

\section{From complete Euclidean remainders to Hankel determinants}\label{sec:euclid-to-hankel}

\begin{proof}[Proof of \Cref{thm:euclid-hankel}]
Assume that
\[
\mathrm{CERS}\in \mathrm{PAC}\text{-}\mathrm{AC}^{0}_{F},
\]
and restrict the assumed circuit family to the input-degree pair
\[
(m,r)=(2n,2n-1).
\]
Recall that
\[
F_n(x)=x^{2n}
\]
and
\[
G_a(x)
=
a_0x^{2n-1}+a_1x^{2n-2}+\cdots+a_{2n-2}x+1.
\]
On the basic open set
\[
D(a_0)=\{a\in F^{2n-1}:a_0\neq 0\},
\]
define the monic rational normalization
\[
\widehat{G}_a:=a_0^{-1}G_a.
\]
Thus \(\widehat{G}_a\) is monic of degree \(2n-1\) at every point of
\(D(a_0)\).

Substitute \(F_n\) for the first input and \(\widehat{G}_a\) for the
second input of the assumed CERS circuit.  This substitution introduces one
division gate computing \(a_0^{-1}\), followed by parallel multiplication gates 
computing the coefficients of \(a_0^{-1}G_a\). Hence it increases the size
only polynomially and the depth only by a constant.  Since \(a_0^{-1}\) is
defined throughout \(D(a_0)\), and since the assumed CERS circuit is defined
on every monic input pair of degrees \(2n\) and \(2n-1\), the resulting
circuit is a formal piecewise rational circuit on \(D(a_0)\).  Its output is
the complete padded ordinary Euclidean remainder sequence of
\[
(F_n,\widehat{G}_a).
\]

Let \(U_n\subseteq D(a_0)\) be the nonempty Zariski-open set from
Lemma~6.2.  By Lemma~6.3, scaling the second input by the nonzero scalar
\(a_0^{-1}\) preserves the relevant subresultant nonvanishing conditions.
Consequently, throughout \(U_n\), the ordinary Euclidean remainder sequence
\[
R_0=F_n,\qquad
R_1=\widehat{G}_a,\qquad
R_2,\ldots,R_n
\]
is normal through \(R_n\).  Write
\[
c_i:=\operatorname{lc}(R_i)
\qquad (1\leq i\leq n).
\]
Normality gives
\[
\deg R_i=2n-i
\qquad (0\leq i\leq n).
\]
By the output convention of Section~2.4, each \(c_i\) is therefore a fixed
coordinate of the \(i\)-th remainder slot.  Its position depends only on
\(n\) and \(i\), and not on the input point \(a\in U_n\).

Applying Proposition~5.4 with \(m=2n\) and \(i=n\) gives
\begin{equation}
\operatorname{psc}^{\,2n,2n-1}_{n}
\bigl(F_n,\widehat{G}_a\bigr)
=
c_n\prod_{i=1}^{n-1}c_i^2.
\label{eq:psc-remainder-product}
\end{equation}
On the other hand, Lemma~6.3, applied with \(k=n\) and
\(\lambda=a_0^{-1}\), gives
\[
\operatorname{psc}^{\,2n,2n-1}_{n}
\bigl(F_n,\widehat{G}_a\bigr)
=
a_0^{-n}
\operatorname{psc}^{\,2n,2n-1}_{n}
\bigl(F_n,G_a\bigr).
\]
By the monomial--Hankel identity~\eqref{eq:psc-hankel},
\[
\operatorname{psc}^{\,2n,2n-1}_{n}
\bigl(F_n,G_a\bigr)
=
(-1)^{\binom{n}{2}}
\Delta_n(a_0,\ldots,a_{2n-2}).
\]
Combining the preceding two identities with~\eqref{eq:psc-remainder-product} yields, throughout
\(U_n\),
\begin{equation}
\Delta_n(a_0,\ldots,a_{2n-2})
=
(-1)^{\binom{n}{2}}
a_0^n c_n\prod_{i=1}^{n-1}c_i^2.
\label{eq:hankel-from-remainders}
\end{equation}
For \(n=1\), the product in~\eqref{eq:hankel-from-remainders} is empty and is interpreted as \(1\).

The right-hand side of~\eqref{eq:hankel-from-remainders} is computed from the input coordinate
\(a_0\) and the fixed remainder-output coordinates
\(c_1,\ldots,c_n\) with constant additional depth and polynomial size.
Indeed, the squares
\[
c_1^2,\ldots,c_{n-1}^2
\]
are computed in parallel.  A single unbounded-fan-in multiplication gate
then receives
\[
(-1)^{\binom{n}{2}},\qquad
c_n,\qquad
c_1^2,\ldots,c_{n-1}^2,
\]
together with \(n\) copies of the input \(a_0\).  The repeated wires carrying
\(a_0\) produce the factor \(a_0^n\).  Thus the additional circuit has
\(O(n)\) wires and constant depth.

We have therefore constructed a scalar-output formal piecewise rational
circuit \(\Phi_n\), of polynomial size and constant depth, that is defined
on \(D(a_0)\) and agrees with
\[
\Delta_n(a_0,\ldots,a_{2n-2})
\]
on the nonempty open subset \(U_n\).  Restrict \(\Phi_n\) to \(U_n\) and
apply Lemma~2.4.  Since \(F\) has characteristic zero, it is infinite, so
the lemma applies and produces an ordinary rational arithmetic circuit
\(C_n\) of no larger size or depth such that
\[
C_n
=
\Delta_n
\]
identically in the rational function field
\[
F(a_0,\ldots,a_{2n-2}).
\]

Finally, \(\Delta_n\) is a polynomial of total degree \(n\).  The
depth-preserving division-elimination theorem of Andrews and Wigderson
\cite[Theorem~2.13]{AW} therefore converts \(C_n\) into a
division-free arithmetic circuit of polynomial size and constant depth.
Consequently,
\[
(\Delta_n)_{n\geq 1}\in \mathrm{AC}^{0}_{F}.
\]
This proves the theorem.
\end{proof}

\section{Proof of the main theorem}\label{sec:main-proof}

\begin{proof}[Proof of \Cref{thm:main}]
By \cref{thm:euclid-hankel,thm:hankel-lb},
$\CERS\notin\PAC^0_{\F}$.  The complete remainder list is a coordinate
projection of the canonically defined extended output $\CEEA$.  Coordinate
projection does not increase size or depth, so
$\CEEA\notin\PAC^0_{\F}$ as well.
\end{proof}

\section{Further Consequences and Complexity-Theoretic Implications}

The argument above yields lower bounds for two further classical
outputs associated with the Euclidean algorithm. The first is the
complete polynomial continued-fraction expansion of a rational
function. The second is the complete profile of fixed-bound principal
subresultant coefficients. In both cases, the full output reveals the
same middle Hankel determinant used in the proof of the main theorem.

\begin{definition}[Complete quotient scheme]
For a monic pair
\[
(f,g)\in\mathcal{M}_{m,r},
\]
let
\[
R_0=f,
\qquad
R_1=g,
\qquad
R_{i-1}=Q_iR_i+R_{i+1}
\]
be its ordinary Euclidean sequence, with the final exact division
written as
\[
R_{t-1}=Q_tR_t.
\]
Define
\[
\operatorname{CQS}_{m,r}(f,g)
:=
\operatorname{Pad}_m(Q_1,\ldots,Q_t),
\]
using the padding convention of Definition~2.6. Equivalently, this
output is the complete polynomial continued-fraction expansion
\[
\frac{f}{g}
=
Q_1+
\cfrac{1}{
Q_2+
\cfrac{1}{
\ddots+
\cfrac{1}{Q_t}
}
}.
\]
\end{definition}

\begin{theorem}[Complete polynomial continued fractions]
\label{thm:continued-fraction-lower-bound}
Let \(F\) be a field of characteristic zero. The complete quotient
scheme is not in
\[
\mathrm{PAC}\text{-}\mathrm{AC}^{0}_{F}.
\]
In particular, the complete polynomial continued-fraction expansion
of a rational function cannot be computed by polynomial-size,
constant-depth piecewise arithmetic circuits.
\end{theorem}

\begin{proof}
Assume, toward a contradiction, that the maps
\[
\operatorname{CQS}_{m,r}
\]
have polynomial-size, constant-depth piecewise arithmetic circuits.

Use the special inputs from Sections~6 and~7. Thus
\[
F_n(x)=x^{2n},
\]
\[
G_a(x)
=
a_0x^{2n-1}
+a_1x^{2n-2}
+\cdots
+a_{2n-2}x
+1,
\]
and, on the basic open set \(D(a_0)\), put
\[
\widehat{G}_a(x)
=
a_0^{-1}G_a(x).
\]
Let
\[
U_n\subseteq D(a_0)
\]
be the nonempty Zariski-open set from Lemma~6.2. As explained in
Section~7, scaling the second input by \(a_0^{-1}\) preserves the
relevant normality conditions. Hence, on \(U_n\), the Euclidean
sequence
\[
R_0=F_n,
\qquad
R_1=\widehat{G}_a,
\qquad
R_2,\ldots,R_n
\]
is normal through \(R_n\). Write
\[
c_i:=\operatorname{lc}(R_i)
\qquad
(0\leq i\leq n).
\]
Since both input polynomials are monic,
\[
c_0=c_1=1.
\]

For \(1\leq i\leq n\), normality gives
\[
\deg R_{i-1}=2n-i+1
\]
and
\[
\deg R_i=2n-i,
\]
so \(Q_i\) has degree one. Let
\[
q_i:=[x]Q_i=\operatorname{lc}(Q_i).
\]
The coefficient \(q_i\) is a fixed coordinate of the \(i\)-th
quotient slot. Comparing leading coefficients in
\[
R_{i-1}=Q_iR_i+R_{i+1}
\]
gives
\[
c_{i-1}=q_ic_i.
\]
All the \(c_i\) are nonzero on \(U_n\), and consequently all the
\(q_i\) are nonzero there. Iterating the preceding identity yields
\[
c_i
=
\left(
\prod_{j=1}^{i}q_j
\right)^{-1}
\qquad
(1\leq i\leq n).
\]

Equation \eqref{eq:hankel-from-remainders} of \Cref{sec:euclid-to-hankel} states that, throughout \(U_n\),
\[
\Delta_n(a_0,\ldots,a_{2n-2})
=
(-1)^{\binom{n}{2}}
a_0^n c_n
\prod_{i=1}^{n-1}c_i^2.
\]
Substituting the expressions for the \(c_i\) gives
\[
\Delta_n(a_0,\ldots,a_{2n-2})
=
(-1)^{\binom{n}{2}}
a_0^n
\prod_{j=1}^{n}
q_j^{-\left(2(n-j)+1\right)}.
\]
Indeed, \(q_j^{-1}\) occurs once through \(c_n\), and it occurs twice
through each of
\[
c_j,c_{j+1},\ldots,c_{n-1}.
\]

The right-hand side is computable from the complete quotient output
by an ordinary rational circuit of polynomial size and constant
depth. For example, each inverse \(q_j^{-1}\) is computed in parallel,
and the required powers are formed by feeding polynomially many
copies to an unbounded-fan-in multiplication gate. All the divisions
are defined on \(U_n\).

Therefore, the assumed circuits for
\[
\operatorname{CQS}_{2n,2n-1},
\]
after the rational substitution
\[
G_a\longmapsto\widehat{G}_a
\]
and the displayed postprocessing, give formal piecewise rational
circuits that agree with \(\Delta_n\) on \(U_n\).

Lemma~2.4 removes the select gates and produces polynomial-size,
constant-depth ordinary rational circuits computing \(\Delta_n\)
identically. Since \(\Delta_n\) has degree \(n\), constant-depth
division elimination gives division-free circuits of polynomial size
and constant depth for the consecutive Hankel determinants. This
contradicts Theorem~1.2.
\end{proof}

\begin{definition}[Complete principal-subresultant profile]
For fixed degree bounds \(m>r\), define
\[
\operatorname{PSP}_{m,r}(f,g)
:=
\bigl(
\operatorname{psc}^{m,r}_0(f,g),
\operatorname{psc}^{m,r}_1(f,g),
\ldots,
\operatorname{psc}^{m,r}_r(f,g)
\bigr),
\]
where the fixed-bound principal subresultant coefficients are those
of Definition~5.1. We call
\[
\operatorname{PSP}
\]
the complete principal-subresultant profile.
\end{definition}

\begin{theorem}[Complete principal-subresultant profiles]
\label{thm:subresultant-profile-lower-bound}
Let \(F\) be a field of characteristic zero. The complete
principal-subresultant profile is not in
\[
\mathrm{PAC}\text{-}\mathrm{AC}^{0}_{F},
\]
even when the input is restricted to monic pairs of consecutive
degree bounds.
\end{theorem}

\begin{proof}
Assume, toward a contradiction, that the maps
\[
\operatorname{PSP}_{m,r}
\]
have polynomial-size, constant-depth piecewise arithmetic circuits on
monic inputs.

Again take
\[
F_n(x)=x^{2n}
\]
and
\[
\widehat{G}_a(x)=a_0^{-1}G_a(x)
\]
on \(D(a_0)\). The pair
\[
\bigl(F_n,\widehat{G}_a\bigr)
\]
is monic and has degree bounds
\[
(2n,2n-1).
\]
Extract the coordinate with index \(n\) from
\[
\operatorname{PSP}_{2n,2n-1}
\bigl(F_n,\widehat{G}_a\bigr).
\]

By Lemma~6.3, scaling the second input by \(a_0^{-1}\) gives
\[
\operatorname{psc}^{2n,2n-1}_n
\bigl(F_n,\widehat{G}_a\bigr)
=
a_0^{-n}
\operatorname{psc}^{2n,2n-1}_n(F_n,G_a).
\]
By the monomial--Hankel identity~\eqref{eq:psc-hankel},
\[
\operatorname{psc}^{2n,2n-1}_n(F_n,G_a)
=
(-1)^{\binom{n}{2}}
\Delta_n(a_0,\ldots,a_{2n-2}).
\]
Consequently, throughout \(D(a_0)\),
\[
\Delta_n(a_0,\ldots,a_{2n-2})
=
(-1)^{\binom{n}{2}}
a_0^n
\operatorname{psc}^{2n,2n-1}_n
\bigl(F_n,\widehat{G}_a\bigr).
\]

Thus the assumed profile circuit, composed with the rational monic
normalization and followed by one fixed coordinate projection and
constant-depth multiplication, gives a polynomial-size,
constant-depth formal piecewise rational circuit agreeing with
\(\Delta_n\) on the nonempty Zariski-open set \(D(a_0)\).

Lemma~2.4 removes the select gates, and constant-depth division
elimination then gives polynomial-size, constant-depth division-free
circuits for \(\Delta_n\). This again contradicts Theorem~1.2.
\end{proof}

The second theorem highlights a useful contrast. A single resultant
is only one extremal member of the subresultant data, whereas the
complete profile contains the middle coefficients that encode the
hard Hankel family. More generally, any canonical complete
subresultant output from which the principal subresultant
coefficients are obtained by a fixed coordinate projection inherits
the same lower bound.

\subsection*{Pad\'e approximation}

We finally record a lower bound for a standard Pad\'e-approximation
problem.  Unlike the preceding consequences, the reduction below does
not require the complete Euclidean transcript.  A polynomial number of
ordinary normalized subdiagonal Pad\'e approximants, computed in
parallel, already exposes a consecutive Hankel determinant.

\begin{definition}[Normalized subdiagonal Pad\'e approximation]
\label{def:pade}
Let $k\geq 1$, and let
\[
C(z)=c_0+c_1z+\cdots+c_{2k-1}z^{2k-1}.
\]
A normalized $[k-1/k]$ Pad\'e approximant to $C$ is a pair of
polynomials
\[
P(z),Q(z)\in F[z],
\]
such that
\[
\deg P\leq k-1,\qquad
\deg Q\leq k,\qquad
Q(0)=1,
\]
and
\[
Q(z)C(z)-P(z)\equiv 0 \pmod{z^{2k}}.
\]
Write
\[
Q(z)=1+q_1z+\cdots+q_kz^k.
\]

Define
\[
\mathcal P_k
:=
\left\{
(c_0,\ldots,c_{2k-1})\in F^{2k}:
\det(c_{i+j})_{0\leq i,j<k}\neq 0
\right\}.
\]
On $\mathcal P_k$ the normalized $[k-1/k]$ Pad\'e approximant is
uniquely determined.  We denote the corresponding output map by
$\operatorname{Pade}_k$.
\end{definition}

\begin{theorem}[Lower bound for Pad\'e approximation]
\label{thm:pade-lower-bound}
Let $F$ be a field of characteristic zero.  The family of normalized
subdiagonal Pad\'e-approximation maps
\[
(\operatorname{Pade}_k)_{k\geq1}
\]
cannot be computed on the domains $\mathcal P_k$ by polynomial-size,
constant-depth piecewise arithmetic circuits.  In fact, the same lower
bound holds even if the required output consists only of the
denominator $Q$.
\end{theorem}

\begin{proof}
Assume, toward a contradiction, that the normalized denominators of
the $[k-1/k]$ Pad\'e approximants can be computed by polynomial-size,
constant-depth piecewise arithmetic circuits.

Fix $n\geq1$ and introduce variables
\[
a_0,a_1,\ldots,a_{2n-2}.
\]
For integers $s,k\geq0$ with $k\geq1$ and
$s+2k-2\leq2n-2$, put
\[
\Delta_{s,k}
:=
\det(a_{s+i+j})_{0\leq i,j<k},
\]
and set
\[
\Delta_{s,0}:=1.
\]
Thus
\[
\Delta_{0,n}
=
\Delta_n(a_0,\ldots,a_{2n-2}).
\]

Consider the nonempty Zariski-open set
\[
U_n
:=
\bigcap_{s=0}^{n-1}
D\!\left(\Delta_{s,n-s}\right).
\]
Each polynomial $\Delta_{s,k}$ is nonzero: for example, after setting
$a_{s+k-1}=1$ and all the other variables occurring in that determinant
equal to zero, the corresponding Hankel matrix is the anti-identity
matrix, up to row order, and therefore has nonzero determinant.
Hence \cref{lem:nonvanishing} shows that $U_n$ is nonempty.

For $0\leq s\leq n-1$, put
\[
k:=n-s.
\]
For each $\tau\in\{0,1\}$ consider the truncated power series
\[
C_{s,k}^{(\tau)}(z)
:=
\sum_{r=0}^{2k-2}a_{s+r}z^r
+
\tau z^{2k-1}.
\]
Let
\[
Q_{s,k}^{(\tau)}(z)
=
1+
q_{s,k,1}^{(\tau)}z+\cdots+
q_{s,k,k}^{(\tau)}z^k
\]
be the denominator of its normalized $[k-1/k]$ Pad\'e approximant.

The coefficient equations in degrees
\[
k,k+1,\ldots,2k-1
\]
are
\[
c_{k+r}
+
\sum_{j=1}^{k}q_{s,k,j}^{(\tau)}c_{k+r-j}
=0,
\qquad
0\leq r<k,
\]
where
\[
c_r=a_{s+r}\quad(0\leq r\leq2k-2),
\qquad
c_{2k-1}=\tau.
\]
If the unknowns are ordered as
\[
\bigl(
q_{s,k,k}^{(\tau)},
q_{s,k,k-1}^{(\tau)},
\ldots,
q_{s,k,1}^{(\tau)}
\bigr)^T,
\]
their coefficient matrix is
\[
H_{s,k}
=
(a_{s+i+j})_{0\leq i,j<k}.
\]
Its determinant is $\Delta_{s,k}$, which is nonzero on $U_n$.
Consequently the normalized denominator exists uniquely there.

We now compute its highest coefficient by Cramer's rule.  Let
\[
\widetilde\Delta_{s+1,k}^{(\tau)}
:=
\det
\begin{pmatrix}
a_{s+1} & a_{s+2} & \cdots & a_{s+k}\\
a_{s+2} & a_{s+3} & \cdots & a_{s+k+1}\\
\vdots  & \vdots  &        & \vdots\\
a_{s+k-1} & a_{s+k} & \cdots & a_{s+2k-2}\\
a_{s+k} & a_{s+k+1} & \cdots & \tau
\end{pmatrix},
\]
where, more invariantly, this is the Hankel determinant
\[
\det(c_{1+i+j})_{0\leq i,j<k}
\]
formed from the coefficients of $C_{s,k}^{(\tau)}$.

Replacing the first column of $H_{s,k}$ by the negative right-hand
side and then moving that column to the last position gives
\[
q_{s,k,k}^{(\tau)}
=
(-1)^k
\frac{\widetilde\Delta_{s+1,k}^{(\tau)}}
     {\Delta_{s,k}}.
\]
The two determinants
$\widetilde\Delta_{s+1,k}^{(1)}$ and
$\widetilde\Delta_{s+1,k}^{(0)}$
differ only in their bottom-right entry.  Expanding their difference
with respect to that entry gives
\[
\widetilde\Delta_{s+1,k}^{(1)}
-
\widetilde\Delta_{s+1,k}^{(0)}
=
\Delta_{s+1,k-1}.
\]
Therefore
\[
(-1)^k
\left(
q_{s,k,k}^{(1)}
-
q_{s,k,k}^{(0)}
\right)
=
\frac{\Delta_{s+1,k-1}}{\Delta_{s,k}}.
\tag{12}
\label{eq:pade-hankel-ratio}
\]

Define
\[
\rho_{s,k}
:=
(-1)^k
\left(
q_{s,k,k}^{(1)}
-
q_{s,k,k}^{(0)}
\right).
\]
Applying~\eqref{eq:pade-hankel-ratio} for
\[
(s,k)
=
(0,n),(1,n-1),\ldots,(n-1,1)
\]
and multiplying the resulting identities gives the telescoping product
\[
\prod_{s=0}^{n-1}\rho_{s,n-s}
=
\frac{\Delta_{1,n-1}}{\Delta_{0,n}}
\frac{\Delta_{2,n-2}}{\Delta_{1,n-1}}
\cdots
\frac{\Delta_{n,0}}{\Delta_{n-1,1}}
=
\frac{1}{\Delta_{0,n}},
\]
because $\Delta_{n,0}=1$.  Hence, throughout $U_n$,
\[
\Delta_n(a_0,\ldots,a_{2n-2})
=
\left(
\prod_{s=0}^{n-1}
\rho_{s,n-s}
\right)^{-1}.
\tag{13}
\label{eq:pade-hankel-recovery}
\]

We now use the assumed Pad\'e circuits.  For each
$s=0,\ldots,n-1$, take two copies of the circuit for the normalized
$[n-s-1/n-s]$ Pad\'e denominator, substitute respectively
$C_{s,n-s}^{(0)}$ and $C_{s,n-s}^{(1)}$, and extract the coefficient
of $z^{n-s}$.  All $2n$ copies are evaluated in parallel.
Equation~\eqref{eq:pade-hankel-ratio} then computes the quantities
$\rho_{s,n-s}$ with only constant additional depth.

On $U_n$, all the determinants appearing in the telescoping chain are
nonzero.  Hence all the quantities $\rho_{s,n-s}$ are nonzero there,
and the divisions in~\eqref{eq:pade-hankel-recovery} are defined.
A single unbounded-fan-in multiplication gate, followed by one
division, therefore produces a polynomial-size, constant-depth formal
piecewise rational circuit agreeing with
\[
\Delta_n(a_0,\ldots,a_{2n-2})
\]
on the nonempty Zariski-open set $U_n$.

Applying \cref{lem:generic-extraction} removes all select gates and yields an ordinary
rational arithmetic circuit of polynomial size and constant depth
that computes $\Delta_n$ identically in
\[
F(a_0,\ldots,a_{2n-2}).
\]
Since $\Delta_n$ is a polynomial of degree $n$, the
constant-depth division-elimination theorem of Andrews and Wigderson
then gives a polynomial-size, constant-depth division-free arithmetic
circuit for $\Delta_n$.

This contradicts \cref{thm:hankel-lb}.  Therefore the normalized subdiagonal
Pad\'e-approximation family is not in
$\mathrm{PAC}\text{-}\mathrm{AC}^0_F$.
Since the proof used only the highest-degree coefficient of the
normalized denominator, the same lower bound already holds for the
denominator output alone.
\end{proof}

\paragraph{Complexity-theoretic significance of the continuant.}
We also gives a useful completeness interpretation of
the continuant.  By \cref{thm:effective-continuant-simulation}, every family
$(f_n)\in\mathrm{VP}_e$ satisfies
\[
(f_n)\preceq_{\mathrm{DEG}}(K_n),
\]
with polynomially bounded reindexing.  Conversely,
$(K_n)\in\mathrm{VP}_e$: indeed, using
\[
K_n(x_1,\ldots,x_n)
=
\bigl(Q(x_n)\cdots Q(x_1)\bigr)_{1,1},
\qquad
Q(z)=
\begin{pmatrix}
z&1\\
1&0
\end{pmatrix},
\]
a balanced multiplication tree for the $2\times2$ matrices gives a
division-free arithmetic formula of polynomial size.  Thus the
continuant family is $\mathrm{VP}_e$-complete under the effective
degeneration $\preceq_{\mathrm{DEG}}$.

This completeness has a direct constant-depth interpretation.  By
\cref{lem:degeneration-preserves-depth}, $\preceq_{\mathrm{DEG}}$ preserves polynomial-size
constant-depth arithmetic computation up to an additive constant in
the depth.  Consequently,
\[
(K_n)\in\mathrm{AC}^0_F
\quad\Longleftrightarrow\quad
\mathrm{VP}_e\subseteq\mathrm{AC}^0_F.
\]
Thus the continuant is a canonical representative of the obstruction
to computing general polynomial-size arithmetic formulas in constant
depth: a constant-depth upper bound for the continuant would collapse
all of $\mathrm{VP}_e$ into $\mathrm{AC}^0_F$, while any
constant-depth lower bound for a family in $\mathrm{VP}_e$ transfers,
through the completeness reduction, to the continuant.  \Cref{thm:continuant-lower-bound}
shows that such a collapse does not occur over fields of characteristic
zero.

\section{Conclusion}

We have shown that the complete extended Euclidean scheme cannot be
computed by polynomial-size, constant-depth piecewise arithmetic
circuits.  In fact, the lower bound already holds for the complete
ordinary remainder sequence.  The proof proceeds by transferring
constant-depth hardness from the continuant to consecutive Hankel
determinants and then recovering a middle Hankel determinant from fixed
coordinates of the Euclidean remainder sequence.

The argument also yields lower bounds for two further complete outputs
associated with the Euclidean algorithm: the complete polynomial
continued-fraction expansion and the complete profile of fixed-bound
principal subresultant coefficients.  These results emphasize a common
feature of the proof: while several individual Euclidean invariants,
such as the gcd or the resultant, admit piecewise constant-depth
algorithms, sufficiently rich complete output data can retain the
constant-depth hardness encoded by intermediate Hankel determinants.

It would be interesting to identify other classical algebraic
algorithms for which a similar distinction occurs between computing a
single final invariant and recovering a complete sequence of
intermediate structural data.  More generally, one may ask which
structured polynomial families can serve as intermediates for
transferring arithmetic constant-depth lower bounds to exact algebraic
output problems.

\end{document}